\documentclass[final,5p,times,twocolumn]{elsarticle}

\usepackage{amsmath,amssymb,amsfonts}
\usepackage{algorithm}
\usepackage{algpseudocode}
\usepackage{graphicx}
\usepackage{textcomp}
\usepackage{xcolor}
\usepackage{caption}
\usepackage{subcaption}
\usepackage{comment}
\usepackage{url}
\usepackage{balance}
\usepackage{hyperref}
\usepackage{mathtools}
\usepackage{multirow}
\usepackage{booktabs}
\usepackage{soul}
\usepackage{colortbl}
\usepackage{pifont}% http://ctan.org/pkg/pifont
\newcommand{\cm}{\ding{51}}%
\newcommand{\xm}{\ding{55}}%
\usepackage{amsthm}
\usepackage{todonotes}
\usepackage{xifthen}
\usepackage{xspace}
\usepackage{balance}
\usepackage{paralist}
\usepackage{silence}
\usepackage[frozencache,cachedir=.]{minted} % beautiful source code
\setminted{fontsize=\fontsize{7}{8.0},tabsize=2,numbers=left,framesep=0pt,numbersep=4pt,highlightcolor=hi_color!20}

\newtheorem{lemma}{Lemma}
\newtheorem{theorem}{Theorem}

\definecolor{Paired-1}{RGB}{31,120,180}
\definecolor{Paired-2}{RGB}{166,206,227}
\definecolor{Paired-3}{RGB}{51,160,44}
\definecolor{Paired-4}{RGB}{178,223,138}
\definecolor{Paired-5}{RGB}{227,26,28}
\definecolor{Paired-6}{RGB}{251,154,153}
\definecolor{Paired-7}{RGB}{255,127,0}
\definecolor{Paired-8}{RGB}{253,191,111}
\definecolor{Paired-9}{RGB}{106,61,154}
\definecolor{Paired-10}{RGB}{202,178,214}
\definecolor{Paired-11}{RGB}{177,89,40}
\definecolor{Paired-12}{RGB}{255,255,153}
\definecolor{Dark2-1}{RGB}{27,158,119}
\definecolor{Dark2-2}{RGB}{217,95,2}
\definecolor{Dark2-3}{RGB}{117,112,179}
\definecolor{Dark2-4}{RGB}{231,41,138}
\definecolor{Dark2-5}{RGB}{102,166,30}
\definecolor{Dark2-6}{RGB}{230,171,2}
\definecolor{Dark2-7}{RGB}{166,118,29}
\definecolor{Dark2-8}{RGB}{102,102,102}

\DeclareRobustCommand{\hlsta}[1]{{\sethlcolor{Paired-5!15}\hl{#1}}}
\DeclareRobustCommand{\hlless}[1]{{\sethlcolor{Paired-7!15}\hl{#1}}}

\newcommand{\otac}[1][]{\ifthenelse{\isempty{#1}}{OTAC\xspace}{OTAC~(#1)\xspace}}
\newcommand{\fertac}{FERTAC\xspace}   % -> First Efficient Resources for TAsk Chains
\newcommand{\twocatac}{2CATAC\xspace} % -> Two-Choice Allocation for TAsk Chains
\newcommand{\herad}{HeRAD\xspace}     % -> Heterogeneous Resource Allocation using Dynamic Programming
\newcommand{\tasks}{\mathcal{T}}
\newcommand{\solution}{S}
\newcommand{\tbig}{\mathcal{B}}
\newcommand{\tlittle}{\mathcal{L}}
\newcommand{\msol}{\mathsf{S}}
\newcommand{\spu}{StreamPU\xspace}
\newcommand{\MIPP}{\textsc{Mipp}\xspace}

\newcommand{\CFS}{\textit{CFS}\xspace}
\newcommand{\EEVDF}{\textit{EEVDF}\xspace}
\newcommand{\schedext}{\textit{sched\_ext}\xspace}
\newcommand{\LAVD}{\textit{LAVD}\xspace}
\newcommand{\BPFLand}{\textit{BPFLand}\xspace}

\DeclareMathOperator*{\dist}{dist}

\newcommand{\StaL}[3][]{\ifthenelse{\isempty{#1}}
                                   {(#2,#3_\mathcal{L})}
                                   {\ifthenelse{\equal{#3}{\string 1}}
                                               {% true case
                                                 $\hlsta{$(#2,#3_\mathcal{L})$}$%
                                               }
                                               {% false case
                                                 $\hlless{$(#2,#3_\mathcal{L})$}$%
                                               }
                                   }
                        }

\newcommand{\StaB}[3][]{\ifthenelse{\isempty{#1}}
                                   {(#2,#3_\mathcal{B})}
                                   {\ifthenelse{\equal{#3}{\string 1}}
                                               {% true case
                                                 $\hlsta{$(#2,#3_\mathcal{B})$}$%
                                               }
                                               {% false case
                                                 $\hlless{$(#2,#3_\mathcal{B})$}$%
                                               }
                                   }
                        }

\journal{Parallel Computing}

\begin{document}

\begin{frontmatter}

\title{Energy-Aware Scheduling Strategies for Partially-Replicable Task Chains\\
  on Heterogeneous Processors}

\author[lip6]{Yacine Idouar}
\ead{yacine.idouar@lip6.fr}
\author[lip6]{Adrien Cassagne}
\ead{adrien.cassagne@lip6.fr}
\author[inria]{Laércio L. Pilla}
%\ead{laercio.pilla@inria.fr}
\author[lip6]{Julien Sopena}
%\ead{julien.sopena@lip6.fr}
\author[lip6]{Manuel Bouyer}
%\ead{manuel.boyer@lip6.fr}
\author[inria]{Diane Orhan}
%\ead{diane.orhan@inria.fr}
\author[lip6]{\\Lionel Lacassagne}
%\ead{lionel.lacassagne@lip6.fr}
\author[lip6]{Dimitri Galayko}
%\ead{dimitri.galayko@lip6.fr}
\author[inp]{Denis Barthou}
%\ead{denis.barthou@enseirb-matmeca.fr}
\author[ubdx]{Christophe Jégo}
%\ead{christophe.jego@ims-bordeaux.fr}

\affiliation[lip6]{organization={LIP6, Sorbonne University, CNRS, UMR7606},
  city={Paris},
  country={France}}

\affiliation[inria]{organization={University of Bordeaux, CNRS, Inria, LaBRI,
    UMR5800},
  city={Talence},
  country={France}}

\affiliation[inp]{organization={Bordeaux INP},
  city={Talence},
  country={France}}

\affiliation[ubdx]{organization={University Bordeaux, CNRS, Bordeaux INP, IMS,
    UMR5218},
  city={Talence},
  country={France}}

%% Abstract
\begin{abstract}
The arrival of heterogeneous (or hybrid) multicore architectures has brought new
performance trade-offs for applications, and efficiency opportunities to
systems. They have also increased the challenges related to thread scheduling,
as tasks' execution times will vary depending if they are placed on big
(performance) cores or little (efficient) ones. In this paper, we focus on the
challenges heterogeneous multicore processors bring to partially-replicable task
chains, such as the ones that implement digital communication standards in
Software-Defined Radio (SDR). Our objective is to maximize the throughput of
these task chains while also minimizing their power consumption. We model this
problem as a pipelined workflow scheduling problem using pipelined and
replicated parallelism on two types of resources whose objectives are to
minimize the period and to use as many little cores as necessary. We propose two
greedy heuristics (FERTAC and 2CATAC) and one optimal dynamic programming
(HeRAD) solution to the problem. We evaluate our solutions and compare the
quality of their schedules (in period and resource utilization) and their
execution times using synthetic task chains. We also study an open source
implementation of the DVB-S2 communication standard based on the StreamPU
runtime. Leading processor vendors are covered with ARM, Apple, AMD, and Intel
platforms. Both the achieved throughput and the energy consumption are
evaluated. Our results demonstrate the benefits and drawbacks of the different
proposed solutions. On average, FERTAC and 2CATAC achieve near-optimal
solutions, with periods that are less than 10\% worse than the optimal (HeRAD)
using fewer than 2 extra cores. These three scheduling strategies now enable
programmers and users of StreamPU to transparently make use of heterogeneous
multicore processors and achieve a throughput that differs from its theoretical
maximum by less than 6\% on average. On the DVB-S2 receiver, it is also shown
that the heterogeneous solutions outperform the best homogeneous ones in terms
of energy efficiency by 8\% on average.
\end{abstract}

%% Keywords

\end{frontmatter}

\section{Introduction}\label{sec:intro}

Multicore processor architectures composed of different types of cores (also
known as \textit{heterogeneous}, \textit{hybrid}, or \textit{asymmetric}) are
increasingly common nowadays. What may have started on low power processors with
ARM’s big.LITTLE architecture in 2011~\cite{randhawa2013software}, has now
become present in processors produced by Apple (since 2020), Intel (since
2021)~\cite{rotem2022intel}, and AMD (since 2023). A common feature in these
processors is an ISA shared between the high-performance (or \textit{big}) and
the high-efficiency (or \textit{little}) cores, which enables the execution of
an application in both types of cores transparently.

Heterogeneous multicore architectures have multiple advantages, such as
providing the opportunity to save energy by turning off big cores when
unnecessary (for battery or environmental reasons). They have also been shown to
outperform homogeneous architectures under a fixed budget (be it area, power, or
both)~\cite{kumar2003single,rodrigues2011performance}. We invite the reader to
check the survey by Mittal on these processors~\cite{mittal2016survey} for more
information. These advantages come with the drawback of higher complexity when
programming parallel applications for these architectures, as one has to decide
how to balance the workload between different cores and core types. In other
words, a solution that does not take care of their differences can result in
lower performance and higher energy consumption.

In this context, we focus on the special characteristics of a kind of parallel
application composed of partially-replicable task chains. A task chain
represents a sequence of tasks with linear dependencies (i.e., all tasks have
one successor and one predecessor, except the source and sink tasks). A task
chain is partially-replicable when only some of its tasks can be replicated
(i.e., multiple copies of the same task can work on different datasets at the
same time to improve throughput). These task chains are found, for instance, in
the context of digital communication standards in Software-Defined Radio
(SDR)~\cite{cassagne2023spu}. We consider the problem of scheduling these
streaming task chains on heterogeneous multicore processors to optimize their
throughput and power consumption in a transparent manner, reducing complexity
for programmers and end users.

A preliminary version of this work appeared in the 34\textsuperscript{th}
Heterogeneity in Computing Workshop (HCW 2025) as part of the
39\textsuperscript{th} IEEE International Parallel and Distributed Processing
Symposium~\cite{orhan2025a}. In it, we (I) provide a formulation of the
optimization problem; (II) propose two greedy heuristics and one optimal dynamic
programming solution; and (III) evaluate the performance of these solutions
using simulation and a real-world digital communication standard
(DVB-S2~\cite{dvbs2}). In this paper, we expand the preliminary
work~\cite{orhan2025a} with the following contributions:

\begin{itemize}
  \item We illustrate the proposed scheduling algorithms through comprehensive
    examples;
  \vspace{-0.2cm}
  \item We provide a proof of optimality of the dynamic programming solution;
  \vspace{-0.2cm}
  \item We give additional information about the \spu runtime used for the
    real-world experimentation, in particular about its internal implementation;
  \vspace{-0.2cm}
  \item We consider two additional platforms based on ARM and AMD processors;
  \vspace{-0.2cm}
  \item We evaluate the power consumption and energy efficiency of the different
    solutions;
  \vspace{-0.2cm}
  \item We measure the impact of different thread placement (pinning) strategies
    on throughput;
  \vspace{-0.2cm}
  \item We implement a new baseline based on a thread-per-block schedule managed
    by the operating system.
\end{itemize}

The remaining sections of this paper are organized as follows:
Section~\ref{sec:rw} discusses related works. Section~\ref{sec:def} presents a
formulation of the problem. Sections~\ref{sec:heuristics} and \ref{sec:dp}
detail our proposed scheduling solutions. Section~\ref{sec:sim} covers
simulations and analytic results while Section~\ref{sec:runtime} introduces
important runtime features. Finally, Section~\ref{sec:sdr} provides real-world
experiments on the DVB-S2 SDR receiver and Section~\ref{sec:conclusion}
concludes this paper.

\section{Related Work}\label{sec:rw}

Our main interest lies in the problem of throughput optimization for
partially-replicable task chains. We focus on solutions using pipeline and
replicated parallelism, and interval mapping~\cite{benoit2013asurvey}.
Pipelining improves throughput by enabling the execution of different dependent
tasks over different datasets simultaneously, adding parallelism to an otherwise
sequential execution. Meanwhile, replication improves throughput by letting
multiple copies of a task work on different datasets at the same time. We
provide more information about these forms of parallelism and interval mapping
in Section~\ref{sec:def}.

In a preliminary version of this work~\cite{orhan2025a}, we have proposed three
solutions to this problem, and evaluated their performance using simulation and
a real-world digital communication standard. In this work, we extend this
evaluation to new platforms, new baselines, new power consumption and energy
efficiency measurements, and new thread pinning strategies.

As we are unaware of any other solutions to our specific research problem in the
state of the art, we will focus our discussion here on variations of this
problem.

\textbf{Throughput on homogeneous architectures:}
\otac~\cite{orhan2025b} provides an optimal solution for partially-replicable
task chains using pipeline and replicated parallelism. We provide more details
about \otac in Section~\ref{sec:heuristics}, as our two greedy heuristics are
based on its main ideas. \otac itself is inspired by Nicol's
algorithm~\cite{nicol1994rectilinear,pinar2004fast}, which is an optimal
solution for the Chain-to-chains partitioning (CCP) problem where only
pipelining is possible. Finally, when all tasks are replicable, the optimal
solution in homogeneous resources is to build a pipeline with a single stage
that is replicated across all resources~\cite{benoit2010complexity}.
Nonetheless, this does not apply for heterogeneous architectures. If the goal is
to improve pipeline load balancing while increasing throughput,
Moreno~\cite{moreno2012} proposed an optimal mapping strategy based on freeing
up resources, by gathering fastest stages together, and replicating the slowest
stages with them. Not only load balancing is not our purpose, but this would
demand to have replicable tasks only.

\textbf{Throughput on heterogeneous architectures:}
Benoit and Robert offered three heuristics for building interval mappings on
unrelated heterogeneous architectures~\cite{benoit2008}. Among them, BSL and BSC
use a combination of binary search and greedy allocation, which is similar to
the general scheme of \otac and our proposed heuristics. These heuristics,
however, do not consider replicated parallelism.

\textbf{Makespan on heterogeneous architectures:}
Topcuoglu et al.~\cite{topcuoglu99} introduced HEFT (one of the most used
heuristics for this kind of problem) and the CPOP to schedule directed acyclic
graphs (DAGs) over unrelated heterogeneous resources. Eyraud-Dubois and
Krumar~\cite{dubois2020} proposed HeteroPrioDep to schedule DAGs over two types
of unrelated resources. Sadly, none of these strategies applies for throughput
optimization, nor for pipeline and replicated parallelism. Agullo et
al.~\cite{agullo2015} studied the performance of dynamic schedulers on two types
of unrelated resources through simulation and real-world experiments. We also
employ both kinds of experiments in our evaluation, but dynamic schedulers from
current runtime systems are usually inefficient at our task granularity of
interest (tens to thousands of $\mu$s)~\cite{slaughter2020}. Benavides et
al.~\cite{benavides2014} proposed a heuristic for the flow shop scheduling
problem on unrelated resources, but their solution is not easily transposable
for pipeline and replicated parallelism.

\textbf{SDR on heterogeneous architectures:} Mack et al.~\cite{mack2022gnu}
proposed the use of the CEDR heterogeneous runtime system to encapsulate and
enable GNU~Radio’s signal processing blocks (tasks) in FPGA and GPU-based
systems on chip. They use dynamic scheduling heuristics and imitation learning
to co-schedule GNU~Radio’s blocks with other applications. In contrast, our
approaches build static pipeline decompositions and schedules for a lower
runtime overhead. We believe our algorithms can be integrated to GNU~Radio in
its future version (4.0)~\cite{morman2022thefuture} when it will abandon its
thread-per-block schedule, enabling better performance by avoiding its current
issues related to locality and OS scheduling
policies~\cite{bloessl2019benchmarking}.

\section{Problem Definition}\label{sec:def}

% initial points about the problem
The problem of maximizing the throughput of a task chain over two kinds of
resources can be modeled as a pipelined workflow scheduling
problem~\cite{benoit2013asurvey}. The workflow can be described as a linear
chain of $n$ tasks $\tasks{} = \{\tau_1,\ldots,\tau_n\}$, meaning $\tau_i$ can
only execute after $\tau_{i-1}$. Tasks are partitioned into two subsets
$\tasks{}_{rep}$ and $\tasks{}_{seq}$ for \textit{replicable} (stateless) and
\textit{sequential} (stateful) tasks. Sequential tasks cannot be replicated due
to their internal state (i.e., replication leads to false results).

The computing system is composed of two types of unrelated heterogeneous
resources $v \in \{\tbig{}, \tlittle{}\}$ representing \textit{big} and
\textit{little} cores, respectively. Big cores are assumed to have the highest
power consumption. The system counts with $b$ big and $l$ little fully-connected
cores. Hereafter, the following notation is used to characterize system
resources: $R = (b,l)$. A task $\tau_i$ has a computation weight (i.e., its
latency) $w^v_{i}$ that depends on the core type $v$. These weights are obtained
by profiling the tasks' execution while running on the different cores.

The mapping strategy on our system is known as interval
mapping~\cite{benoit2013asurvey}, where $\tasks{}$ is partitioned into $k$
contiguous intervals.
We call the $i^{th}$ interval in the format $[\tau_c,
\tau_e]$ ($c\leq e$) a \textit{stage} noted as $s_i$. A stage is defined as
replicable if it contains only replicable tasks. We define $r_i$ and $v_i$ as
the number and the type of resources dedicated to $s_i$, respectively.
A stage is statically mapped to these resources.
The weight of a stage $s$ with $r$ cores of type $v$ is defined in
Equation~\eqref{eq:wsp}.
% Be careful, when $r > 1$, the weight of a stage differs from its latency.
Other characteristics, such as the communication weights between tasks and
network bandwidth are considered out of the scope of our current work due to our
focus on heterogeneous \textit{multicore architectures} (keeping data exchanges
local) and interval mapping (which minimizes data transfers).
\begin{equation}\label{eq:wsp}
  w(s, r, v) =
  \begin{cases}
  \sum_{\tau\in s} w^v_{\tau} & \text{if}~s\cap \tasks{}_{seq} \not=
  \emptyset,~r \geq 1, \\
  \frac{1}{r}\sum_{\tau\in s} w^v_{\tau} & \text{if}~s\cap \tasks{}_{seq} =
  \emptyset,~r \geq 1, \\
  \infty & \text{otherwise}
  \end{cases}
\end{equation}

Our \textbf{main objective} is to find, before the execution of the application,
a solution $\solution{} = (\mathsf{s}, \mathsf{r},\mathsf{v})$ with
$\mathsf{s}=(s_1, \ldots, s_k)$ that maximizes throughput.

As throughput is inversely proportional to the period, we will refer
to this problem as a \textbf{period minimization} problem in the remaining
sections. The period of a solution $\mathsf{P}(\solution{})$ is given by the
greatest weight among all stages (Equation~\eqref{eq:tnp}). A solution is only
valid if the number of available resources is respected
(Equation~\eqref{eq:res}).
\begin{equation}\label{eq:tnp}
  \mathsf{P}(\mathsf{s},\mathsf{r},\mathsf{v}) =
  \max_{i \in [1,k]} w(s_i,r_i,v_i)
\end{equation}
\begin{equation}\label{eq:res}
    \sum_{v_i = \tbig{}} r_i \leq b, ~~~\sum_{v_i = \tlittle{}} r_i \leq l
\end{equation}

Our \textbf{secondary objective} is to minimize the power consumption of the
solution (as minimizing energy makes no sense when dealing with a continuous
data stream). Solutions to this problem depend on the information available.
For instance, if the power consumed by each task in each core type were
available (and independent from other tasks in the same core), using formulas
such as the ones given McGough~\cite{mcgough2022}, the objective of period
minimization could be prioritized over the power one. A second option would be
to assume a fixed power consumption per core used of each type. We chose to work
with a different proxy: the use of little cores instead of big ones, as they
have lower power consumption. In this case, our secondary objective is to
\textbf{use as many little cores as necessary} (and not more) to achieve the
minimum period. We will see how this impacts our proposed algorithms next.

\section{New Greedy Heuristics: \fertac and \twocatac}\label{sec:heuristics}

We propose two heuristics to schedule partially-replicable task chains on two
types of resources. They are both based on \otac~\cite{orhan2025b} which is able
to find optimal solutions for homogeneous resources. In a nutshell, \otac uses a
binary search to set up a target period (similar to Algorithm~\ref{algo:solve})
and then tries to greedily build a schedule by packing as many tasks as possible
in each stage (as in Algorithm~\ref{algo:stage}). Our new heuristics, named
\fertac and \twocatac, use different means to minimize the period while
using as many little cores as necessary. We discuss the main ideas behind them
next.

\subsection{\fertac}\label{subsec:dwotac}

\textit{First Efficient Resources for TAsk Chains}, or \textbf{\fertac} for
short, aims to use little cores to build each stage. Big cores are only used
when it is not possible to respect the target period. We will explain how
\fertac operates by first covering its methods common to \twocatac (i.e.,
Schedule and ComputeStage) and then discussing its specific implementation of
the ComputeSolution method (Algorithm~\ref{algo:dwotac}).

In  order to help with the legibility of Schedule and ComputeStage, some common
support methods are listed in Algorithm~\ref{algo:support}.
\textbf{IsValid} (line~1) checks if a partial or complete solution respects the
constraints of the problem (it has one or more stages, it respects the target
period, and it does not surpass the numbers of big or little cores).
\textbf{MaxPacking} (line~3) returns the end index for the stage (starting from
task $s$) that includes the most tasks without surpassing the target period
given a number and type of core.
\textbf{RequiredCores} (line~5) returns the number of cores of a given type that
a stage requires to respect the target period.
\textbf{IsRep} (line~6) informs if a stage is replicable (i.e., only contains
replicable tasks), and \textbf{FinalRepTask} (line~7) returns the index of the
last replicable task in a contiguous sequence.

\textbf{Schedule} (Algorithm~\ref{algo:solve}) follows a binary search procedure
similar to the one used for the CCP problem~\cite{pinar2004fast}. It sets the
lower period bound by the maximum between (I) replicating all tasks over all
resources and (II) the sequential task with the largest weight
(line~1)\footnote{For the sake of simplicity, we assume here that tasks run
fastest on big cores. This only affects the computation of period bounds and can
easily be changed with some more comparisons between weights.}. The upper period
bound is based on the minimum period plus the largest task weight (line~2). The
binary search (lines~5--14) tries to find a solution with a target middle
period. If the solution is valid, we store it and update the upper bound
(lines~9--10), or else we update the lower bound (line~12).
The search stops when the difference between the bounds is smaller than an
epsilon (line~3). Its value is defined as $\frac{1}{b+l}$ to account for
replicated stages that can lead to periods with values that are not integers
($\frac{1}{r}$ in Equation~\eqref{eq:wsp}), but whose fractions cannot go
smaller than the number of resources. In total, this requires
$\mathcal{O}(\log (w_{max}(b+l)))$ calls to ComputeSolution, with $w_{max} =
\max_{\tau \in \tasks} w^{\tlittle{}}_{\tau}$.

%%% SUPPORT METHODS
\begin{algorithm}[!ht]
\caption{Common support methods}
\label{algo:support}
\scriptsize
\begin{algorithmic}[1]
\State \textbf{IsValid}$((\mathsf{s},\mathsf{r},\mathsf{v}),b,l,P):$
\State ~~~~\textbf{return} $(|\mathsf{s}| > 0$ ~\textbf{and}~$\mathsf{P}(\mathsf{s},\mathsf{r},\mathsf{v}) \leq P$ ~\textbf{and}~$\sum_{i \in [1,|\mathsf{v}|] \wedge v_i = \tbig{}} r_i\leq b$ ~\textbf{and}~$\sum_{i \in [1,|\mathsf{v}|] \wedge v_i = \tlittle{}} r_i \leq l)$

\State \textbf{MaxPacking}$(\tasks{},   s, c, v, P) :$
\State ~~~~\textbf{return} $\max (s, \max_{i \in [s, |\tasks{}|]} \{i~|~w([\tau_{s},\tau_{i}],c,v) \leq P\})$

\State \textbf{RequiredCores}$(\tasks{}, s, e, v, P) : \lceil \frac{w([\tau_{s},\tau_{e}],1,v)}{P} \rceil$

\State \textbf{IsRep}$(\tasks{}, s, e) : [\tau_{s},\tau_{e}] \cap \tasks{}_{seq} = \emptyset$

\State \textbf{FinalRepTask}$(\tasks{},   s, e) : \max_{i \in [e, |\tasks{}|]} \{i~|~$ IsRep$(\tasks{}, s, i)\}$
\end{algorithmic}
\end{algorithm}

%%% SCHEDULE
\begin{algorithm}[!ht]
\caption{Schedule (common method)}
\label{algo:solve}
\scriptsize
\begin{algorithmic}[1]
\Require Set of tasks $\tasks{}$, big cores $b$, little cores $l$.
\Ensure Pipelined and replicated solution $\solution{}_{best}$.
\State $P_{min} \leftarrow \max ( \frac{\sum_{\tau\in \tasks{}} w^{\tbig{}}_{\tau}}{b+l} , \max_{\tau \in \tasks{}_{seq}} w^{\tbig{}}_{\tau} )$ \Comment{Minimum expected period}
\State $P_{max} \leftarrow P_{min} + \max_{\tau \in \tasks} w^{\tlittle{}}_{\tau}$ \Comment{Maximum expected period}
\State $\epsilon \leftarrow \frac{1}{b+l}$
\State $\solution{}_{best} \leftarrow \emptyset$
\While{$P_{max} - P_{min} \geq \epsilon$}
    \State $P_{mid} = \frac{P_{max} + P_{min}}{2}$ \Comment{Target period for the binary search iteration}
    \Statex $\triangleright$~\textcolor{blue}{\textit{ComputeSolution is different for \fertac (Algorithm~\ref{algo:dwotac}) and \twocatac (Algorithm~\ref{algo:notac})}}
    \State $\solution{} \leftarrow$~ComputeSolution$(\tasks{},1,b,l,P_{mid})$%~\textit{Solution depending on the algorithm.}
    \If{IsValid$(\solution{}, b, l, P_{mid})$} \Comment{Checks for validity (Algorithm~\ref{algo:support})}
        \State $\solution{}_{best} \leftarrow \solution{}$ \Comment{New best solution}
        \State $P_{max} \leftarrow \mathsf{P}(\solution{})$ \Comment{Can only decrease the target period from here}%\Comment{$\solution{} = (\mathsf{s},\mathsf{r},\mathsf{v})$}
    \Else
        \State $P_{min} \leftarrow P_{mid}$ \Comment{Can only increase the target period}
    \EndIf
\EndWhile
\State \textbf{return} $\solution{}_{best}$
%\Statex $\triangleright$~\textit{Resource from }
\end{algorithmic}
\end{algorithm}

%%% COMPUTE STAGE
\begin{algorithm}[!ht]
\caption{ComputeStage (common method)}
\label{algo:stage}
\scriptsize
\begin{algorithmic}[1]
\Require Set of tasks $\tasks{}$, task index $s$, cores $c$, core type $v$,
  target period $P$.
\Ensure Task index $e$, used cores $u$.
%\State $n \leftarrow |\tasks{}|$
% First packing
\State $e \leftarrow $~MaxPacking$(\tasks{},   s, 1, v, P)$ \Comment{Packs tasks using one core (Algorithm~\ref{algo:support})}
\State $u \leftarrow $~RequiredCores$(\tasks{}, s, e, v, P)$ \Comment{Cores needed for this stage (Algorithm~\ref{algo:support})}
\If{$e \neq n$~\textbf{and} IsRep$(\tasks{}, s, e)$} \Comment{If the stage is replicable (Algorithm~\ref{algo:support})}
    % Stage is stateless, more tasks can be added
    \State $e \leftarrow $~FinalRepTask$(\tasks{},   s, e)$ \Comment{Extends the stage (Algorithm~\ref{algo:support})}
    \State $u \leftarrow $~RequiredCores$(\tasks{}, s, e, v, P)$
    \If{$u > c$} \Comment{Not enough cores for all tasks, needs to reduce the stage}
        % We do not have enough resources for these tasks, so we remove tasks from the stage
        \State $e \leftarrow $~MaxPacking$(\tasks{},   s, c, v, P)$~;~$u \leftarrow c$
    % We have enough resources for these tasks
    \ElsIf{$e \neq n$} \Comment{Checks if it is better to leave one core for the next stage}
        % There exists more tasks after this stage // Checks if the resources can be better distributed
        \State $f \leftarrow $~MaxPacking$(\tasks{},   s, u-1, v, P)$
        \If{RequiredCores$(\tasks{}, f+1, e+1, v, P) = 1$}
            % Better to leave some stateless tasks for the next stage
            \State $e \leftarrow f$~;~$u \leftarrow u-1$ \Comment{Best to reduce the stage}
        \EndIf
    \EndIf
\EndIf
\State \textbf{return} $e, u$
\end{algorithmic}
\end{algorithm}

\textbf{ComputeStage} (Algorithm~\ref{algo:stage}) tries to find where to finish
a stage and how many cores (of a given type) are required to respect the target
period. It first tries to pack as many tasks as possible in the stage using a
single core (line~1). We check how many cores the stage requires for the case
where the last task in the chain is replicable and its weight surpasses the
target period (line~2). If the stage is replicable (line~3), it is extended to
include all following replicable tasks (line~4). If this long stage requires
more cores than available, it is reduced to respect the target period
(lines~5--7). If this stage is not the final one, it means there is a sequential
task after it. We check if it is better to move this stage’s final tasks to the
next stage while saving one core and, if that is the case, we update the end of
the stage (lines~9--12). All these tests guarantee that the stage is packing as
many tasks as possible with the given cores.

\fertac's \textbf{ComputeSolution} recursively computes a solution for a given
target period (Algorithm~\ref{algo:dwotac}) by first trying to build a stage
with little cores (line~1), and only moving to big cores if no valid solution
was found (lines~2--3). If the stage is final, then it finishes the recursion
(lines~8--9). If not, then we are required to continue computing the next stage
with the remaining cores (lines~11--13). ComputeSolution returns the list of
stages\footnote{The operation $\cdot$ is used for the concatenation of new items
at the start of the stages, resources, and core types lists.} if a valid
solution is found (line~15).

%%% COMPUTE SOLUTION (PER ALGORITHM)
\begin{algorithm}
\caption{ComputeSolution for \fertac}
\label{algo:dwotac}
\scriptsize
\begin{algorithmic}[1]
\Require Set of tasks $\tasks{}$, task index $s$, big cores $b$, little cores
  $l$, target period $P$.
\Ensure Pipelined and replicated [partial] solution.
\State $e, u \leftarrow $~ComputeStage$(\tasks{}, s, l, \tlittle{}, P)$~;~$v \leftarrow \tlittle{}$ \Comment{Uses little cores (Algorithm~\ref{algo:stage})}
\If{\textbf{not} IsValid$(([\tau_s,\tau_e],u,v), b, l, P)$} \Comment{Checks for validity (Algorithm~\ref{algo:support})}
    \State $e, u \leftarrow $~ComputeStage$(\tasks{}, s, b, \tbig{}, P)$~;~$v \leftarrow \tbig{}$ \Comment{Needed to use big cores}
    \If{\textbf{not} IsValid$(([\tau_s,\tau_e],u,v), b, l, P)$} \Comment{No valid solution for both cases}
        \State \textbf{return} $(\emptyset, \emptyset, \emptyset)$
    \EndIf
\EndIf
% Stage is valid and the last one
\If{$e = |\tasks{}|$}
    \State \textbf{return} $([\tau_s,\tau_e], u, v)$ \Comment{Returns the valid, final stage}
\Else \Comment{Needs to continue building stages}
% Stage is valid and there are more to compute
    \State $b \leftarrow b-u$~\textbf{if}~$v = \tbig{}$ \Comment{Updates available cores for next stages}
    \State $l \leftarrow l-u$~\textbf{if}~$v = \tlittle{}$
    \State $(\mathsf{s},\mathsf{r},\mathsf{v}) \leftarrow$~ComputeSolution$(\tasks{},e+1,b,l,P)$ \Comment{Computes the next stages}
    \If{IsValid$((\mathsf{s},\mathsf{r},\mathsf{v}), b, l, P)$}
        \State \textbf{return} $([\tau_s,\tau_e] \cdot \mathsf{s}, u \cdot \mathsf{r}, v \cdot
 \mathsf{v})$ \Comment{Returns the list of stages}
    \Else
        \State \textbf{return} $(\emptyset, \emptyset, \emptyset)$
    \EndIf
\EndIf
\end{algorithmic}
\end{algorithm}

\begin{figure}[!ht]
  \centering
  \includegraphics[width=0.9\columnwidth]{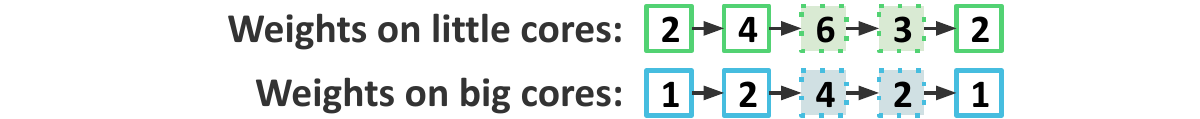}
  \caption{Example of a task chain with five tasks. Full, transparent (resp.
    dashed, shaded) boxes represent sequential (resp. replicable) tasks. Green
    (resp. blue) boxes indicate the weights for tasks in little (resp. big)
    cores.}
  \label{fig:example_chain}
\end{figure}

\begin{figure}[!ht]
  \centering
  \includegraphics[width=0.9\columnwidth]{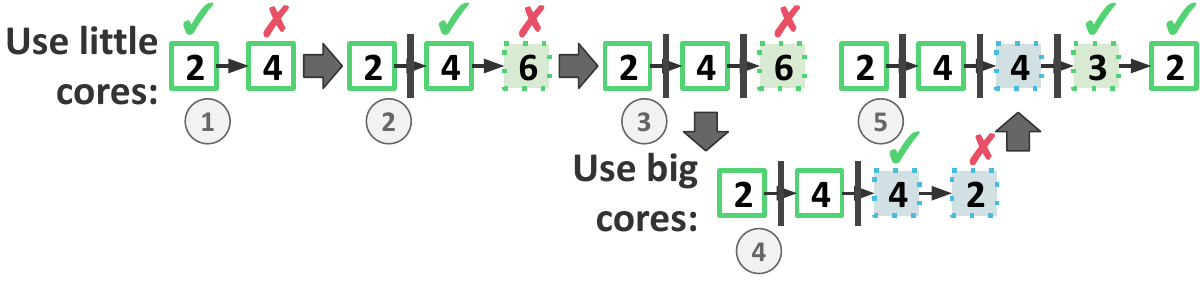}
  \caption{Solution found with \fertac’s ComputeSolution for the task chain from
    Figure~\ref{fig:example_chain} for 3 big and 3 little cores, and a target
    period of $P=5$.}
  \label{fig:fertac}
\end{figure}

Figure~\ref{fig:fertac} illustrates the five steps (gray circles) followed by
Algorithm~\ref{algo:dwotac} for the task chain illustrated in
Figure~\ref{fig:example_chain} in a scenario with 3 big cores, 3 little cores,
and a target period of $P=5$. In step~1, \fertac finds a valid stage with the
first task by itself using one little core. It then recursively calls
ComputeSolution (line 13) with task index 2 and one less little core in order to
compute the next stages. \fertac finds another valid stage with a single task
using one little core in step~2. With the little core that remains, it fails to
find a valid stage including the third task in step~3. This happens because the
task surpasses the target period and there are not enough little cores left to
replicate the stage. \fertac then tries to build the stage using big cores
(line~3) in step~4. It decides the best option is to create a stage with the
third task only and a big core. This decision is made after checking that
leaving the third and fourth task together using two big cores (lines~4 and~5 in
Algorithm~\ref{algo:stage}) would be worse than leaving the fourth task to be
scheduled with the last task on another stage (lines~9 to~12 in
Algorithm~\ref{algo:stage}). Finally, in step~5, \fertac builds the last stage
using the last two remaining tasks and a little core. It then returns this valid
stage (line~9 in Algorithm~\ref{algo:dwotac}) to start completing the recursion.
The preceding method calls check that the returned stages are valid in line~14
and return the list of stages (including their own) in line~15.

Regarding the complexity of \fertac, multiple implementation aspects have to be
considered. Given $n = |\tasks{}|$, we chose to precompute the sum of weights
for any given stage using two prefix sums in $\mathcal{O}(n)$. We also chose to
precompute if any stage is replicable (Algorithm~\ref{algo:support}, line~6) in
$\mathcal{O}(n^2)$ for simplicity, but this cost could be amortized by
sequentially checking each task (as in \otac~\cite{orhan2025b}). The validity of
a solution (Algorithm~\ref{algo:support}, lines~1--2) has its cost amortized by
checking each stage as it is built. Packing tasks in a stage and identifying the
final replicable task in a sequence can be done task by task. Finally, it can be
seen that each task is considered a constant number of times in ComputeStage
(Algorithm~\ref{algo:stage}), and a task can only be considered for two stages
in sequence (and twice for both types of cores). With all these aspects taken
into consideration, our implementation of \fertac requires $\mathcal{O}(n\log
(w_{max}(b+l)) + n^2)$ operations.

\subsection{\twocatac}\label{subsec:notac}

While \fertac tries to use little cores as soon as possible, \textit{Two-Choice
Allocation for TAsk Chains} (or \textbf{\twocatac} for short) tries both types
of cores for building a stage at each time. This enables the strategy to make
better use of little cores in later stages, and to potentially consider
different secondary objectives when comparing solutions. This comes at the cost
of an exponential increase in the number of solutions to check.

%%% COMPUTE SOLUTION (PER ALGORITHM)
\begin{algorithm}[t!]
\caption{ComputeSolution for \twocatac}
\label{algo:notac}
\scriptsize
\begin{algorithmic}[1]
\Require Set of tasks $\tasks{}$, task index $s$, big cores $b$, little cores
  $l$, target period $P$.
\Ensure Pipelined and replicated [partial] solution $\solution{}_{best}$.

\For{$v \in \{\tlittle{},\tbig{}\}$} \Comment{Builds solution for this stage with both types of cores}
    \State $r \leftarrow b$~\textbf{if}~$v = \tbig{}$~\textbf{else}~$l$
    \State $e_v, u_v \leftarrow $~ComputeStage$(\tasks{}, s, r, v, P)$ \Comment{Greedily builds a stage (Algorithm~\ref{algo:stage})}
    \If{\textbf{not} IsValid$(([\tau_s,\tau_{e_v}],u_v,v), b, l, P)$} \Comment{Checks for validity (Algorithm~\ref{algo:support})}
        \State $\solution{}_v \leftarrow (\emptyset, \emptyset, \emptyset)$ \Comment{No valid stage with this type of cores}
    \ElsIf{$e_v = |\tasks{}|$}
        \State $\solution{}_v \leftarrow ([\tau_s,\tau_{e_v}], u_v, v)$ \Comment{Valid, final stage option}
    \Else
        \State $b_v \leftarrow b-u_v$~\textbf{if}~$v = \tbig{}$~\textbf{else}~$b$ \Comment{Updates available cores for next stages}
        \State $l_v \leftarrow l-u_v$~\textbf{if}~$v = \tlittle{}$~\textbf{else}~$l$
        \State $(\mathsf{s}_v,\mathsf{r}_v,\mathsf{v}_v) \leftarrow$~ComputeSolution$(\tasks{},e_v+1,b_v,l_v,P)$ \Comment{Next stages}
        \If{IsValid$((\mathsf{s}_v,\mathsf{r}_v,\mathsf{v}_v), b_v, l_v, P)$}
            \State $\solution{}_v \leftarrow([\tau_s,\tau_{e_v}] \cdot \mathsf{s}_v, u_v \cdot \mathsf{r}_v, v \cdot \mathsf{v}_v)$ \Comment{Valid combined solution}
        \Else
            \State $\solution{}_v \leftarrow (\emptyset, \emptyset, \emptyset)$
        \EndIf
    \EndIf
\EndFor
\State \textbf{return} ChooseBestSolution$(\solution_\tbig{}, \solution_\tlittle{}, b, l, P)$ \Comment{Picks the best solution (Algorithm~\ref{algo:choosebest})}
\end{algorithmic}
\end{algorithm}

\twocatac's \textbf{ComputeSolution} (Algorithm~\ref{algo:notac}) computes the
stage for both big and little cores (lines 1--3). In each case, if the final
stage is identified, it is stored for comparison (line~7), or else the recursion
is launched for the next stage (line~11) and combined with the current stage
(line~13).

\begin{algorithm}[t!]
\caption{ChooseBestSolution (part of \twocatac)}
\label{algo:choosebest}
\scriptsize
\begin{algorithmic}[1]
\Require Solutions $\solution_\tbig{}$ and $\solution_\tlittle{}$, big cores
  $b$, little cores $l$, target period $P$.
\Ensure Pipelined and replicated [partial] solution $\solution{}_{best}$.
\If{IsValid$(\solution_\tbig{}, b, l, P)$} \Comment{Checks for validity (Algorithm~\ref{algo:support})}
    % valid for big
    \If{IsValid$(\solution_\tlittle{}, b, l, P)$}
        % has to compare solutions
        \For{$v \in \{\tbig{},\tlittle{}\}$} \Comment{Compares the core usage of the solutions}
            \State $(\mathsf{s},\mathsf{r},\mathsf{v}) \leftarrow \solution{}_{v}$
            \State $\Sigma b_v \leftarrow \sum_{i \in [1,|\mathsf{v}|] \wedge v_i = \tbig{}} r_i$
            \State $\Sigma l_v \leftarrow \sum_{i \in [1,|\mathsf{v}|] \wedge v_i = \tlittle{}} r_i$
        \EndFor
        % Set of comparisons
        \If{$\Sigma l_\tbig{} > \Sigma l_\tlittle{}$~\textbf{and}~$\Sigma b_\tbig{} < \Sigma b_\tlittle{}$}
            \State $\solution{}_{best} \leftarrow \solution_\tbig{}$ \Comment{$\solution_\tbig{}$ makes better usage of little cores}
        \ElsIf{$\Sigma l_\tbig{} < \Sigma l_\tlittle{}$~\textbf{and}~$\Sigma b_\tbig{} > \Sigma b_\tlittle{}$}
            \State $\solution{}_{best} \leftarrow \solution_\tlittle{}$ \Comment{$\solution_\tlittle{}$ makes better usage of little cores}
        \ElsIf{$\Sigma l_\tbig{} + \Sigma b_\tbig{} < \Sigma l_\tlittle{} + \Sigma b_\tlittle{}$}
            \State $\solution{}_{best} \leftarrow \solution_\tbig{}$ \Comment{$\solution_\tbig{}$ uses fewer cores}
        \Else
            \State $\solution{}_{best} \leftarrow \solution_\tlittle{}$ \Comment{$\solution_\tlittle{}$ uses fewer cores}
        \EndIf
    \Else
        % only valid for big
        \State $\solution{}_{best} \leftarrow \solution_\tbig{}$ \Comment{Only valid solution}
    \EndIf
\ElsIf{IsValid$(\solution_\tlittle{}, b, l, P)$}
    % only valid for little
    \State $\solution{}_{best} \leftarrow \solution_\tlittle{}$ \Comment{Only valid solution}
\Else   % No valid solution
    \State $\solution{}_{best} \leftarrow (\emptyset, \emptyset, \emptyset)$ \Comment{No valid solution}
\EndIf
\State \textbf{return} $\solution{}_{best}$
\end{algorithmic}
\end{algorithm}

ComputeSolution employs \textbf{ChooseBestSolution}
(Algorithm~\ref{algo:choosebest}) to compare the solutions for both types of
cores. A solution is directly returned if it is the only valid one (lines~18
and~21). In the other case, the solution that better exchanges big cores for
little ones is returned (lines~9 and~11) or, in the last scenario, the one that
uses fewer cores is chosen (lines~13 and~15). As ComputeSolution's objective is
to find a schedule that respects the target period, there is no need to compare
the stages' weights for the different solutions.

\begin{figure}[!ht]
  \centering
  \includegraphics[width=\columnwidth]{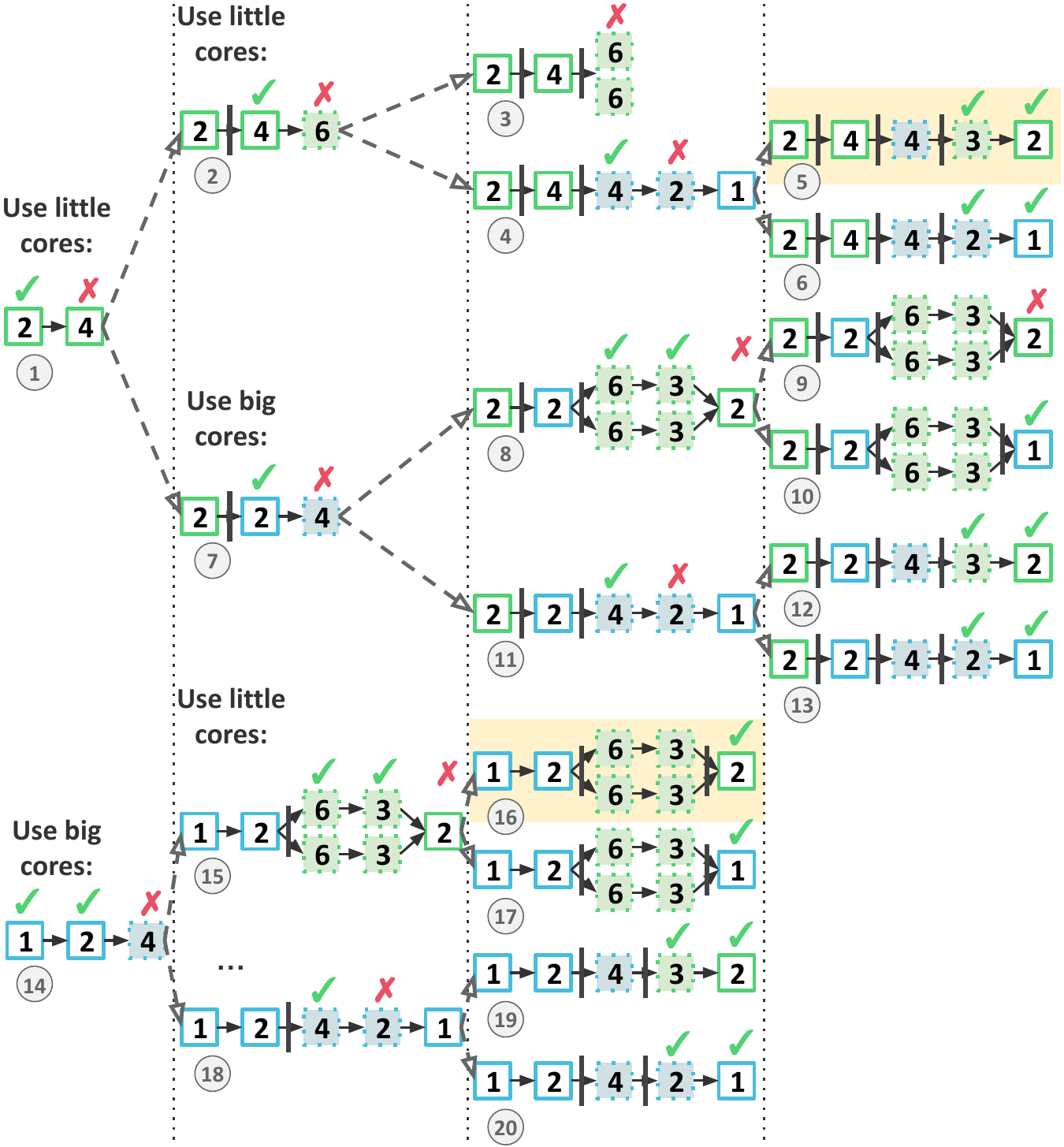}
  \caption{Solutions computed with \twocatac’s ComputeSolution for the task
    chain from Figure~\ref{fig:example_chain} for 3 big and 3 little cores, and
    a target period of $P=5$. The best solutions from each branch (steps 5 and
    16) are highlighted in yellow.}
  \label{fig:2catac}
\end{figure}

Figure~\ref{fig:2catac} illustrates the steps (gray circles) followed by
Algorithm~\ref{algo:notac} for the task chain illustrated in
Figure~\ref{fig:example_chain} in a scenario with 3 big cores, 3 little cores,
and a target period of $P=5$. The recursive process is equivalent to a
depth-first traversal of the tree representing the different stage options.
In the first iteration of the loop (lines~1--18), \twocatac builds a valid stage
with the first task and one little core in step 1. It then makes the recursive
call to ComputeSolution (line 11) to generate the possibilities for the next
stage. When it gets to step 3, it fails to find a valid stage with little cores
(line 4 verifies that it surpasses the number of little cores available), so
\twocatac continues in step 4 by building a stage using big cores. When steps 5
and 6 are finished, \twocatac compares the two versions for the last stage and
returns only the one deemed to be the best (line 19). In this case, it will
return the stage built in step 5 because it better exchanges big cores for
little ones (line 11 in Algorithm~\ref{algo:choosebest}). The final solution
represented in step 5 will be returned through the recursion until it gets back
to the level from step 2, where step 7 will then be launched to start evaluating
solutions using big cores for the second stage of the pipeline. In the end, when
all recursive calls are done, \twocatac will compare the solutions finished in
step~5 and~16, as they each represent the best ones found in each branch of the
starting recursion tree. As both use the same numbers of big and little cores,
the schedule finished in step~5 will be the one returned (line~15 in
Algorithm~\ref{algo:choosebest}).

\twocatac's complexity is defined by its recursion’s possible solutions tree.
Other aspects, such as the comparison between two solutions, are amortized by
capturing relevant information while computing them. For instance, we provide
the accumulated core usages when combining solutions
(Algorithm~\ref{algo:notac}, line~13) instead of computing them each time
(Algorithm~\ref{algo:choosebest}, lines~5--6). Given the considerations
previously discussed for \fertac, we can conclude that \twocatac displays a
worst-case complexity in $\mathcal{O}(2^n\log (w_{max}(b+l)))$ when each stage
contains only one task. This can be prohibitive for larger task chains, but it
is still faster than the optimal solution (to be shown next) in some scenarios
(see Section~\ref{sec:sim}).

\section{Optimal Dynamic Programming Solution}\label{sec:dp}

This section first presents our dynamic programming solution, and later proves
the optimality of the solution. Our dynamic programming solution can be defined
based on this problem's recurrence. Let $\mathsf{P}^*(j,b,l)$ be the best period
achieved when mapping tasks from $\tau_1$ to $\tau_j$ using up to $b$~big cores
and $l$~little cores. $\mathsf{P}^*(j,b,l)$ can be computed using the recurrence
in Equation~\eqref{eq:dp} with $\mathsf{P}^*(0,b,l) = 0$ and
$\mathsf{P}^*(j,0,0) =\infty$ for $j>0$.
\begin{equation}\label{eq:dp}
\begin{split}
    \mathsf{P}^*(j,b,l) = & \\
    \min\limits_{i \in [1,j]} &
    \begin{cases}
    \min\limits_{u \in [1,b]} \biggl( \max \Bigl( \mathsf{P}^*(i-1,b-u,l),
      w([\tau_{i},\tau_{j}],u,\tbig{}) \Bigr) \biggr)  \\
    \min\limits_{u \in [1,l]}~ \biggl( \max \Bigl( \mathsf{P}^*(i-1,b,l-u),
      w([\tau_{i},\tau_{j}],u,\tlittle{}) \Bigr) \biggr)
    \end{cases}
\end{split}
\end{equation}

Equation~\eqref{eq:dp} shows that an optimal solution can be built from partial
optimal solutions. The best solution is found by trying all possible starts for
the stage finishing in $\tau_j$ and all possible resource distributions between
this stage and previous ones for both core types. This recurrence can be
computed in $\mathcal{O}(j^2bl(b+l))$ time and $\mathcal{O}(jbl)$ space.

\begin{algorithm}[!ht]
\caption{\herad}
\label{algo:opt_main}
\scriptsize
\begin{algorithmic}[1]
\Require Set of tasks $\tasks{}$, big cores $b$, little cores $l$.
\Ensure Pipelined and replicated solution $\solution{}_{best}$.
\For{$i \in [1,|\tasks{}|]~,~j \in [0,b]~,~k \in [0,l]$} \Comment{Initializes solution matrix}
    \State $\msol{}_{Pbest}[i][j][k] \leftarrow \infty$ \Comment{Minimal maximum period}
    \State $\msol{}_{prev}[i][j][k]  \leftarrow (0,0)$ \Comment{Big and little cores in the previous stages}
    \State $\msol{}_{acc}[i][j][k]  \leftarrow (0,0)$ \Comment{Accumulated big and little cores}
    \State $\msol{}_{v}[i][j][k]  \leftarrow \tlittle{}$ \Comment{Type of core used in the stage}
    \State $\msol{}_{start}[i][j][k] \leftarrow 1$ \Comment{Index of the starting task of the stage}
\EndFor
\State SingleStageSolution$(1,\msol{},\tasks{},b,l)$ \Comment{Single task in a single stage (Algorithm~\ref{algo:opt_single})}
\For{$e \in [2,|\tasks{}|]$}
    \State SingleStageSolution$(e,\msol{},\tasks{},b,l)$ \Comment{All $e$ tasks in a single stage (Algorithm~\ref{algo:opt_single})}
    \For{$u_b \in [0,b]$} \Comment{Solutions with more than one stage}
        \For{$u_l \in [0,l]$} \Comment{and varying numbers of cores}
            \If{$u_b \neq 0$~\textbf{or}~$u_l \neq 0$}
                \State RecomputeCell$(e,\msol{},\tasks{},u_b,u_l)$ \Comment{$\mathsf{P}^*(e,u_b,u_l)$ (Algorithm~\ref{algo:opt_cell})}
            \EndIf
        \EndFor
    \EndFor
\EndFor
\State \textbf{return} ExtractSolution$(\msol{},\tasks{},b,l)$ \Comment{Converts the matrix to a solution (Algorithm~\ref{algo:opt_extract})}
\end{algorithmic}
\end{algorithm}

%%% SINGLE STAGE SOLUTION
\begin{algorithm}[!ht]
\caption{SingleStageSolution (part of \herad)}
\label{algo:opt_single}
\scriptsize
\begin{algorithmic}[1]
\Require Task index $t$, solution matrix $\msol{}$, set of tasks $\tasks{}$,
  big cores $b$, little cores $l$.
\Ensure Updated $\msol{}[t][:][:]$.
\For{$r_l \in [1,l]$} \Comment{Initializes row with a stage using little cores}
    \State $\msol{}_{Pbest}[t][0][r_l] \leftarrow w([\tau_{1},\tau_{t}],r_l,\tlittle{})$
    \State $\msol{}_{acc}[t][0][r_l] \leftarrow (0,r_l)$~\textbf{if}~IsRep$(\tasks{}, 1, t)$~\textbf{else}~$(0,1)$
\EndFor
\For{$r_b \in [1,b]$}
    \State $w_b \leftarrow w([\tau_{1},\tau_{t}],r_b,\tbig{})$ \Comment{Computes the stage with big cores}
    \State $u_b \leftarrow r_b$~\textbf{if}~IsRep$(\tasks{}, 1, t)$~\textbf{else}~$1$
    \For{$r_l \in [0,l]$} \Comment{Compares if it is better to use $r_b$ big }
        \If{$w_b < \msol{}_{Pbest}[t][0][r_l]$} \Comment{~or $r_l$ little cores for this single stage}
            \State $\msol{}_{Pbest}[t][r_b][r_l] \leftarrow w_b$
            \State $\msol{}_{acc}[t][r_b][r_l] \leftarrow (u_b,0)$
            \State $\msol{}_{v}[t][r_b][r_l] \leftarrow \tbig{}$
        \Else
            \State $\msol{}_{Pbest}[t][r_b][r_l] \leftarrow \msol{}_{Pbest}[t][0][r_l]$
            \State $\msol{}_{acc}[t][r_b][r_l] \leftarrow \msol{}_{acc}[t][0][r_l]$
        \EndIf
    \EndFor
\EndFor
\end{algorithmic}
\end{algorithm}

%%% RECOMPUTE CELL
\begin{algorithm}[!ht]
\caption{RecomputeCell (part of \herad)}
\label{algo:opt_cell}
\scriptsize
\begin{algorithmic}[1]
\Require Task index $j$, solution matrix $\msol{}$, set of tasks $\tasks{}$,
  big cores available $b$, little cores available $l$.
\Ensure Updated $\msol{}[j][b][l]$.
\State $C \leftarrow \msol{}[j][b][l]$ \Comment{Uses the initial solution from SingleStageSolution (Algorithm~\ref{algo:opt_single})}
\State $C \leftarrow$~CompareCells$(C, \msol{}[j][b][l-1])$~\textbf{if}~$l>0$ \Comment{Compares to neighbor solutions}
\State $C \leftarrow$~CompareCells$(C, \msol{}[j][b-1][l])$~\textbf{if}~$b>0$ \Comment{~using one less core}
\For{$i \in [2,j]$~in reverse order} \Comment{External $\min_{i \in [1,j]}$ (Equation~\ref{eq:dp})}
    \For{$u \in [1,b]$} \Comment{Internal $\min_{u \in [1,b]}$ (Equation~\ref{eq:dp})}
        \State $B_{Pbest} \leftarrow \max \left( \msol{}_{Pbest}[i-1][b-u][l], w([\tau_{i},\tau_{j}],u,\tbig{}) \right)$
        \State $(a_b, a_l) \leftarrow \msol{}_{acc}[i-1][b-u][l]$
        \State $B_{acc} \leftarrow (a_b+u, a_l)$~\textbf{if}~IsRep$(\tasks{}, i, j)$~\textbf{else}~$(a_b+1, a_l)$
        \State $B_{prev} \leftarrow (b-u, a_l)$~;~$B_{v} \leftarrow \tbig{}$~;~$B_{start} \leftarrow i$
        \State $C \leftarrow$~CompareCells$(C,B)$ \Comment{Keeps the best solution (Algorithm~\ref{algo:opt_compare})}
    \EndFor
    \For{$u \in [1,l]$} \Comment{Internal $\min_{u \in [1,l]}$ (Equation~\ref{eq:dp})}
        \State $L_{Pbest} \leftarrow \max \left( \msol{}_{Pbest}[i-1][b][l-u], w([\tau_{i},\tau_{j}],u,\tlittle{}) \right)$
        \State $(a_b, a_l) \leftarrow \msol{}_{acc}[i-1][b][l-u]$
        \State $L_{acc} \leftarrow (a_b, a_l+u)$~\textbf{if}~IsRep$(\tasks{}, i, j)$~\textbf{else}~$(a_b, a_l+1)$
        \State $L_{prev} \leftarrow (a_b, l-u)$~;~$L_{v} \leftarrow \tlittle{}$~;~$L_{start} \leftarrow i$
        \State $C \leftarrow$~CompareCells$(C,L)$ \Comment{Keeps the best solution (Algorithm~\ref{algo:opt_compare})}
    \EndFor
\EndFor
\State $\msol{}[j][b][l] \leftarrow C$ \Comment{Stores the best solution}
\end{algorithmic}
\end{algorithm}

%%% COMPARE CELLS
% if ((neigh_load < current_load) ||
%            ((neigh_load == current_load) && (neigh_little > current_little) && (neigh_big < current_big)) ||
%            ((neigh_load == current_load) && (neigh_little <= current_little) && (neigh_big <= current_big))) {
\begin{algorithm}[!ht]
\caption{CompareCells (part of \herad)}
\label{algo:opt_compare}
\scriptsize
\begin{algorithmic}[1]
\Require Matrix cells for partial solutions $C$ (current) and $N$ (new).
\Ensure Best partial solution.
\State $(c_b, c_l) \leftarrow C_{acc}$~;~$(n_b, n_l) \leftarrow N_{acc}$
\If{$(C_{Pbest} > N_{Pbest})$ ~\textbf{or}~$(C_{Pbest} = N_{Pbest}$ \textbf{and} $c_l < n_l$ \textbf{and} $c_b > n_b)$ ~\textbf{or}~$(C_{Pbest} = N_{Pbest}$ \textbf{and} $c_l \geq n_l$ \textbf{and} $c_b \geq n_b)$}
    \State \textbf{return} $N$
\Else
    \State \textbf{return} $C$
\EndIf
\end{algorithmic}
\end{algorithm}

%%% EXTRACT SOLUTION
\begin{algorithm}[!ht]
\caption{ExtractSolution (part of \herad)}
\label{algo:opt_extract}
\scriptsize
\begin{algorithmic}[1]
\Require Solution matrix $\msol{}$, set of tasks $\tasks{}$, big cores $b$, little cores $l$.
\Ensure Pipelined and replicated solution $\solution{}_{best}$.
\State $e \leftarrow |\tasks{}|$~;~$s \leftarrow |\tasks{}|$~;~$r_b \leftarrow b$~;~$r_l \leftarrow l$
\State $(\mathsf{s},\mathsf{r},\mathsf{v}) \leftarrow (\emptyset, \emptyset, \emptyset)$
\While{$e\geq 1$}
    \State $s \leftarrow \msol{}_{start}[e][r_b][r_l]$ \Comment{Start of the stage}
    \State $(u_b, u_l) \leftarrow \msol{}_{acc}[e][r_b][r_l]$
    \State $v \leftarrow \msol{}_{v}[e][r_b][r_l]$ \Comment{Type of core used}
    \State $(p_b, p_l) \leftarrow \msol{}_{prev}[e][r_b][r_l]$
    \If{$s>1$} \Comment{Gets the number of cores used in this stage only}
        \State $(c_b, c_l) \leftarrow \msol{}_{acc}[s-1][p_b][p_l]$
        \State $u_b \leftarrow u_b - c_b$~;~$u_l \leftarrow u_l - c_l$
    \EndIf
    \State $r \leftarrow u_b$ \textbf{if} $v = \tbig{}$ \textbf{else} $u_l$ \Comment{Number of cores used}
    \State $(\mathsf{s},\mathsf{r},\mathsf{v}) \leftarrow ([\tau_s,\tau_e] \cdot \mathsf{s}, r \cdot \mathsf{r},v \cdot \mathsf{v})$ \Comment{Adds the stage to the solution}
    \State $e \leftarrow s-1$~;~$r_b \leftarrow p_b$~;~$r_l \leftarrow p_l$ \Comment{Index for the predecessor stage}
\EndWhile
\State \textbf{return} $(\mathsf{s},\mathsf{r},\mathsf{v})$
\end{algorithmic}
\end{algorithm}

\textbf{\herad}, or short for \textit{Heterogeneous Resource Allocation using
Dynamic programming} (Algorithm~\ref{algo:opt_main}), implements the optimal
strategy of Equation~\eqref{eq:dp} in a bottom-up fashion while also considering
the secondary objective of using as many little cores as necessary. It starts by
initializing a solution matrix $\msol{}$ that will contain all optimal partial
solutions (lines~1--7). It then computes all optimal solutions for the first
task in the chain with all possible numbers of cores (line~8) using
SingleStageSolution (i.e., $\mathsf{P}^*(1,:,:)$). In the next step and for
increasing numbers of tasks (index $e$ in line~9), the algorithm computes a
first solution where all tasks belong to the same stage (line~10). Using the
notation from Equation~\eqref{eq:dp}, these represent the solutions using values
based only on $w([\tau_{1},\tau_{e}],u,\tbig{})$ or  $w([\tau_{1},\tau_{e}],u,
\tlittle{})$. Then, it computes the optimal partial solution for this number of
tasks with increasing numbers of big and/or little cores iteratively (line~14)
using RecomputeCell. This represents computing $\mathsf{P}^*(e,b,l)$ reusing the
values computed before for $\mathsf{P}^*(i,b,l)$ for $i<e$ that are stored in
$\msol{}[i]$. The algorithm finishes by going backwards in the solution matrix
and identifying the stages that belong to the optimal solution (line~19) using
ExtractSolution (Algorithm~\ref{algo:opt_extract}). We have also added an extra
step that merges consecutive stages if they are replicable and using the same
core type. This has no impact in the minimum period achieved, but it leads to
solutions with fewer stages.

\textbf{SingleStageSolution} (Algorithm~\ref{algo:opt_single}) finds the best
solutions when putting all considered tasks in the same stage. It computes and
stores the weight of the stage using increasing numbers of little cores, taking
care to register that sequential stages can only benefit from a single core
(lines~1--4). It then considers an increasing number of big cores (lines~6--7)
and compares their solutions with the ones using little cores (lines~8--17). It
stores the solution with minimum period in the matrix, solving ties in favor of
the little cores (lines~9--16).

\textbf{RecomputeCell} (Algorithm~\ref{algo:opt_cell}) tries all possible
optimal solutions for a scenario with a given number of tasks, big cores and
little cores. It uses the solution from SingleStageSolution as a starting point
and compares it to other solutions with one fewer big or little core
(lines~1--3) that have been previously computed. It then computes all possible
solutions for $\max \mathsf{P}^*(i-1,b-u,l), w([\tau_{i},\tau_{j}],u,\tbig{})$
(lines~4--11) and $\max \mathsf{P}^*(i-1,b,l-u), w([\tau_{i},\tau_{j}],u,
\tlittle{})$ (lines~4 and 12--18), comparing them sequentially to the best
solution found so far, and the best solution is stored in the matrix (line~20).
We implement an optimization that limits comparisons to a single core (instead
of a range of cores in lines 5 and 12) if the stage is sequential. All solution
comparisons make use of \textbf{CompareCells}
(Algorithm~\ref{algo:opt_compare}). It returns the solution with the minimum
maximum period. In the case of ties, the solution that better exchanges big
cores for little cores is returned or, in the last scenario, the one that uses
fewer cores is chosen

\begin{figure}[!ht]
  \centering
  \includegraphics[width=0.9\columnwidth]{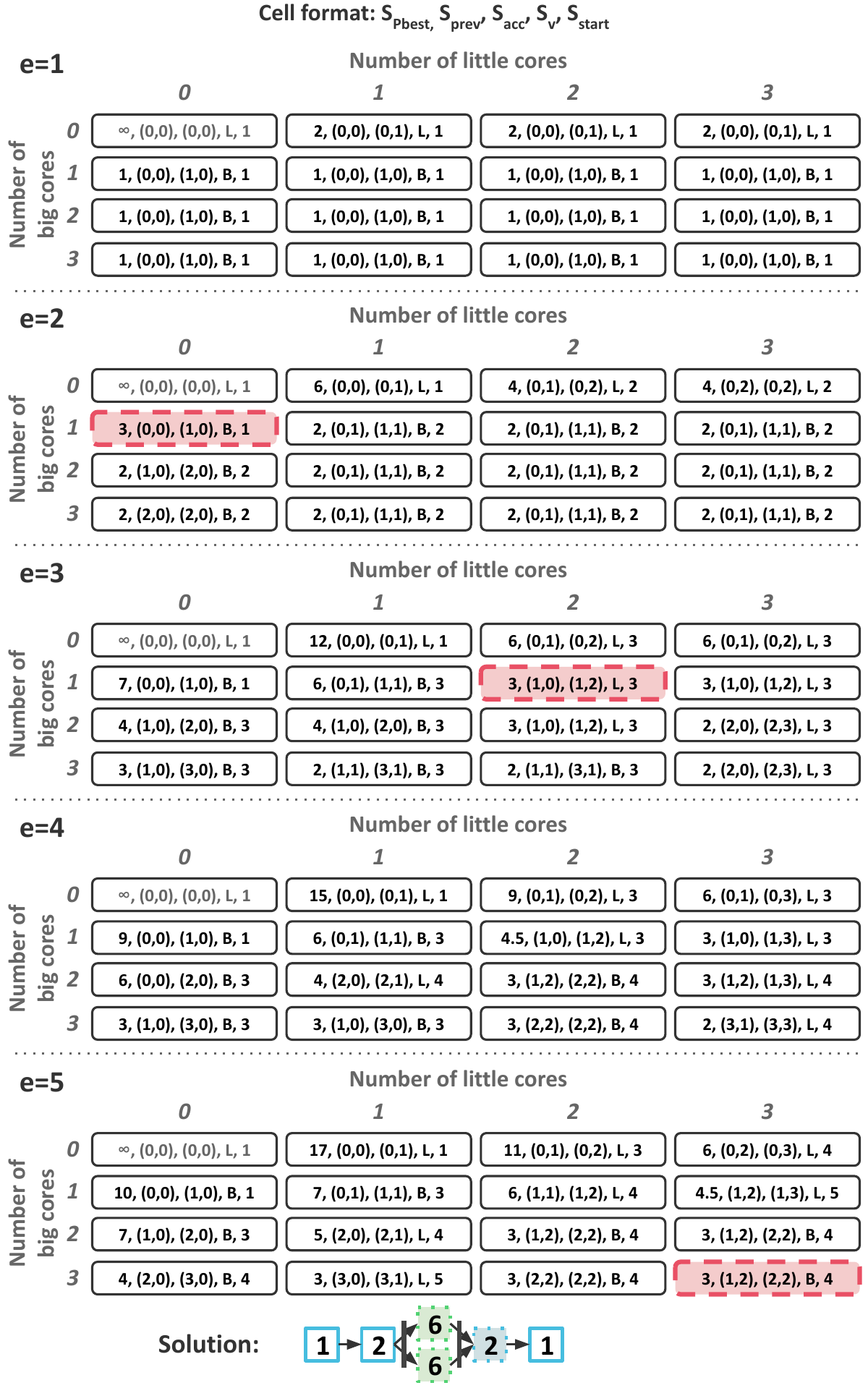}
  \caption{Representation of the values stored in the different 3D matrices by
    \herad to compute an optimal solution for the task chain from
    Figure~\ref{fig:example_chain} for 3 big and 3 little cores. Each 2D table
    represents the values computed for a given task index $i$. Red, dashed,
    shaded cells represent the path taken by Algorithm~\ref{algo:opt_extract} to
    extract the final solution.}
  \label{fig:herad}
\end{figure}

Figure~\ref{fig:herad} illustrates the final state of the solution matrix
$\msol{}$ after the execution of \herad for task chain illustrated in
Figure~\ref{fig:example_chain} in a scenario with 3 big cores, 3 little cores.
Each 2D table represents the values in $\msol{}[e]$ for different numbers of big
and little cores. The value of $\msol{}_{Pbest}$ in the cell represents the
optimal period found for that number of tasks, big cores, and little cores (in
other words, the value of $\mathsf{P}^*$ from Equation~\eqref{eq:dp}). They are
computed in order for increasing values of task index $e$ from top to bottom,
left to right (lines 11--17 in Algorithm~\ref{algo:opt_main}).
Algorithm~\ref{algo:opt_extract} extracts the final solution (illustrated on the
bottom of the figure) by going from the last cell ($\msol{}[5][3][3]$ in this
case) backwards, which ends up covering all red, dashed, shaded cells in
Figure~\ref{fig:herad}.

\subsection{\herad's optimality proof}

The optimality of \herad can be demonstrated by induction\footnote{For the sake
of brevity, we provide only a resumed proof of this solution’s optimality.
Suffice to say, similar proofs have been provided for other interval-based
mapping algorithms~\cite{agrawal2008} and dynamic programming algorithms with
secondary objectives handled when comparing partial
solutions~\cite{nunes2024}.}. Its proof combines elements of
Equation~\eqref{eq:dp} and its implementations in
Algorithms~\ref{algo:opt_main},~\ref{algo:opt_single}, and~\ref{algo:opt_cell}.
At each step, we first cover the period minimization aspect of the solution,
followed by the idea of using as many little cores as necessary.

\begin{lemma}\label{proof:1:1}
  The solution for $\mathsf{P}^*(1,b,l)$ is optimal.
\end{lemma}

\begin{proof}
  The only possible solutions for $\mathsf{P}^*(1,b,l)$ include a single
  pipeline stage using big or little cores. Algorithm~\ref{algo:opt_single} is
  used to compute the solution for $j=1$ (Algorithm~\ref{algo:opt_main},
  line~8). It stores the minimum between the solutions using $b$ big cores or
  $l$ little cores (Algorithm~\ref{algo:opt_single}, lines 5--9), therefore it
  is optimal in period.

  Regarding the use of little cores, the algorithm first computes solutions
  using them (lines 1--4) and then solves ties with big cores in favor of the
  little ones (line~9, use of $<$), thus being optimal in this aspect too.
\end{proof}

\begin{lemma}\label{proof:1:2}
  The solution for $\mathsf{P}^*(j,b,l)$ is optimal if the solutions for
  $\mathsf{P}^*(i,r_b,r_l)$ are also optimal for $i<j$, $r_b\leq b$, and
  $r_l\leq l$.
\end{lemma}

\begin{proof}
  % optimal in period
  The period of $\mathsf{P}^*(j,b,l)$ takes its value from the minimum period
  among all possible starts for the stage finishing in $\tau_j$ using all
  possible resource distributions (Equation~\eqref{eq:dp}, loops in
  Algorithm~\ref{algo:opt_single} using Algorithm~\ref{algo:opt_compare}, and
  Algorithm~\ref{algo:opt_cell}). To consider another schedule with a smaller
  period is a contradiction, as it requires having a suboptimal
  $\mathsf{P}^*(i,r_b,r_l)$, or a value that is smaller than the minimum of all
  possible solutions, so $\mathsf{P}^*(j,b,l)$ is optimal regarding its period.

  % optimal in resources
  Regarding the use of little cores, Algorithm~\ref{algo:opt_compare} always
  solves ties in the benefit of the solution that better exchanges big cores for
  little cores or the one that uses fewer cores. We also ensure that solutions
  having one less big or little core available are propagated from previous
  solutions (Algorithm~\ref{algo:opt_cell} lines 2--3), thus the solution is
  also optimal in this aspect.
\end{proof}

\begin{theorem}\label{proof:final}
  \herad yields optimal solutions regarding the period achieved and the use of
  little cores.
\end{theorem}

\begin{proof}
  Lemmas \ref{proof:1:1} and \ref{proof:1:2} prove the optimality of the base
  case and the inductive step, so \herad is optimal.
\end{proof}

As given by Theorem~\ref{proof:final}, \herad provides schedules with minimal
periods while using as many little cores as necessary with the potential issue
of a high complexity. We next evaluate how its benefits and drawbacks measure
against our greedy heuristics.

\section{Simulations \& Analytic Results}\label{sec:sim}

In this section, we study synthetic task chains and processors to check how well
the strategies are able to optimize our two objectives, and also to profile
their execution time. Comparisons include our three strategies and
\otac~\cite{orhan2025b} (which handles homogeneous resources only). We provide
more details about our experimental environment and results in the next
subsections. Source code, result files, and scripts are freely available
online~\cite{amp-scheduling_v2.0}.

\subsection{Experimental Environment}

Experiments were executed on a Dell \textbf{Latitude~7420} notebook (Intel Core
i7-1185G7 @ 3~GHz, 32~GB LPDDR4 @ 3733~MT/s, 512~GB NVMe SSD) running Linux
(Ubuntu 24.04.1 LTS kernel 6.8.0-51, g++ 13.3.0). For period and core usage
measurements, 1000~task chains of 20~tasks were generated. Task weights were
randomly set in the integer interval [1,100] uniformly for big cores with a
slowdown in the interval [1,5] for little cores (rounded using the ceiling
function). In order to evaluate how the replicable tasks affect the strategies,
the stateless ratio (SR) (i.e., fraction of tasks that are replicable) of each
chain was set equal to $\{0.2, 0.5, 0.8\}$ for different scenarios. We set the
number of big ($b$) and little ($l$) cores (e.g., the resources $R = (b,l)$) in
the simulation using three different pairs $\{(16_\mathcal{B},4_\mathcal{L}),
(10_\mathcal{B},10_\mathcal{L}), (4_\mathcal{B},16_\mathcal{L})\}$.
For the execution time profiling, we generate 50~task chains for varying numbers
of tasks ($20i | i \in [1,8]$), pairs of numbers of cores ($(20i,20i) | i \in
[1,8]$), and SRs.

\subsection{Slowdown Compared to \herad}

Given that \herad always provides minimal periods, we use the slowdown ratio
$\frac{\mathsf{P}(\solution{}_{\text{other}})}
{\mathsf{P}(\solution{}_{\text{\herad}})}$ to compare strategies.
Figure~\ref{fig:densities_zoomed} illustrates the cumulative distributions
(1000~task chains) of slowdown ratios for varying resources and slowdown ratios.
Each line represents a strategy, with \otac[B] (resp. \otac[L]) using only big
(resp. little) cores. We can notice in the first column ($SR=0.2$) that
\twocatac and \fertac tend to find minimal periods in most cases, but they
become less effective as the SR increases (other columns). The higher the SR,
the more likely for the period to be limited by replicable tasks and the higher
the number of replication options to explore, making it harder to find the best
solution.

\begin{figure}[!ht]
  \centering
  \begin{subfigure}{1.0\columnwidth}
    \includegraphics[width=0.32\columnwidth]%
    {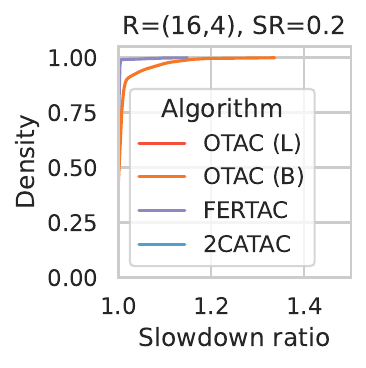}
    \hfill
    \includegraphics[width=0.32\columnwidth]%
    {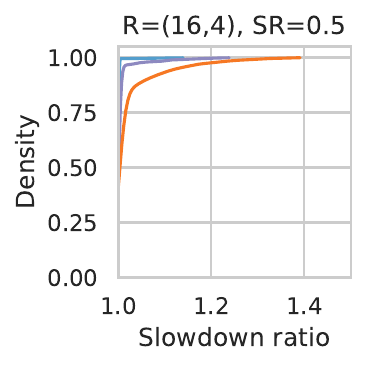}
    \hfill
    \includegraphics[width=0.32\columnwidth]%
    {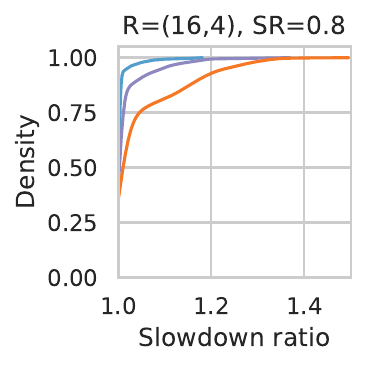}
    \\
    \includegraphics[width=0.32\columnwidth]%
    {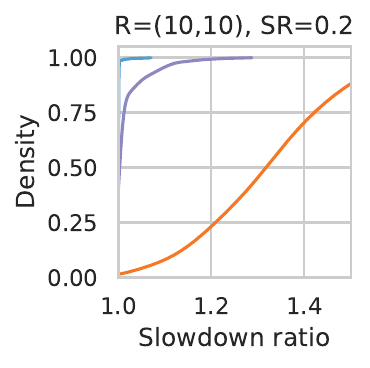}
    \hfill
    \includegraphics[width=0.32\columnwidth]%
    {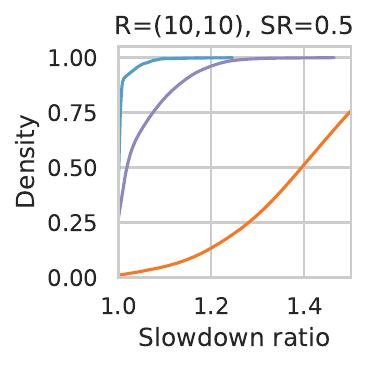}
    \hfill
    \includegraphics[width=0.32\columnwidth]%
    {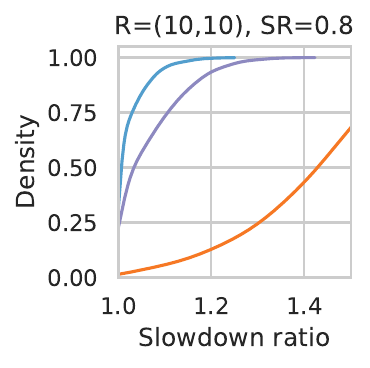}
    \\
    \includegraphics[width=0.32\columnwidth]%
    {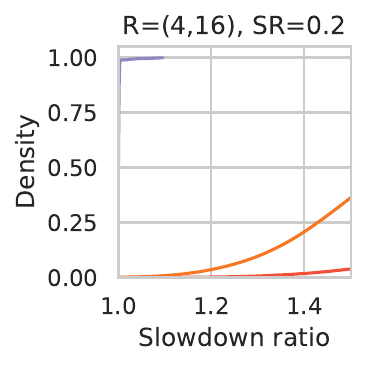}
    \hfill
    \includegraphics[width=0.32\columnwidth]%
    {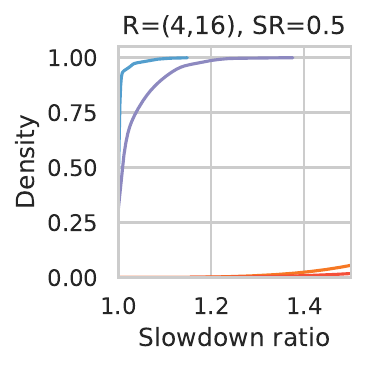}
    \hfill
    \includegraphics[width=0.32\columnwidth]%
    {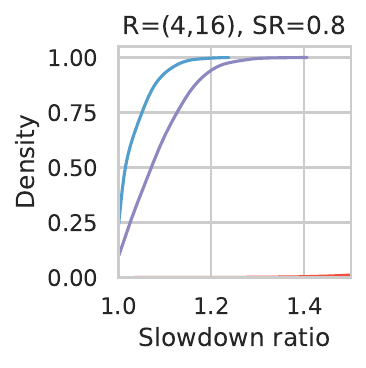}
    \caption{Results zoomed in the slowdown interval $[1,1.5]$.}
    \label{fig:densities_zoomed}
  \end{subfigure}

  \vspace{0.2cm}
  \begin{subfigure}{1.0\columnwidth}
    \includegraphics[width=0.32\columnwidth]%
    {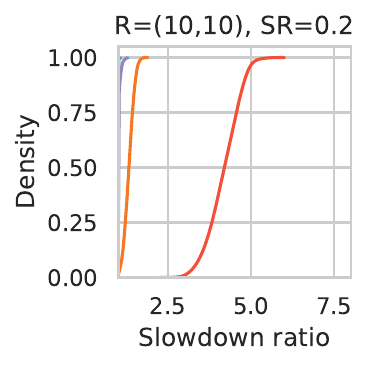}
    \hfill
    \includegraphics[width=0.32\columnwidth]%
    {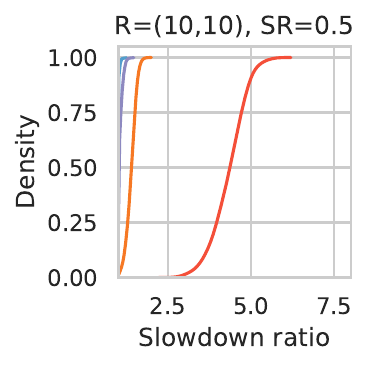}
    \hfill
    \includegraphics[width=0.32\columnwidth]%
    {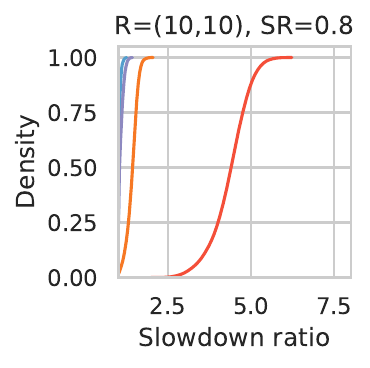}
    \caption{Full slowdown interval for $R=(10_{\mathcal{B}},
      10_{\mathcal{L}})$.}
    \label{fig:densities}
  \end{subfigure}
  \caption{Cumulative solution (e.g., \emph{density}) distributions of slowdown
    ratios (cf. \herad) for different heuristics. Columns represent different
    SRs. In Subfigure~\ref{fig:densities_zoomed}, rows represent different pairs
    of resources.}
  \label{fig:densities_all}
\end{figure}

\otac[B] performs similarly to \fertac only when many big cores are available
(first row). Meanwhile, \otac[L] never finds optimal solutions because it lacks
the big cores to handle the slowest tasks. The gap between these strategies can
be better seen in Figure~\ref{fig:densities} with a full range of slowdown
ratios.

We summarize our simulation statistics in Table~\ref{tab:stats}. When few little
cores are available ($R=(16_{\mathcal{B}},4_{\mathcal{L}})$), \twocatac and
\fertac find the majority of minimal periods, leading to $1\%$ or lower
slowdowns on average. Even for scenarios with different numbers of cores,
\twocatac and \fertac achieve average slowdown ratios limited to $1.03$ and
$1.08$, respectively, which represent $97.1\%$ and $92.6\%$ of the potential
throughput. Their worst results ($R=(10_{\mathcal{B}},10_{\mathcal{L}}),
SR=0.5$) were limited to maximum slowdowns of $1.23$ and $1.41$, respectively
(or $81.3\%$ and $70.1\%$ of the potential throughput). In comparison, \otac[B]
shows average slowdown ratios comparable to the maximum slowdowns of \fertac for
$R=(10_{\mathcal{B}},10_{\mathcal{L}})$, and even worse when even fewer big
cores are available. It emphasizes the importance of using both core types
together.

\begin{table*}[h]
\centering
\caption{Simulation statistics for all scheduling strategies. Each 4-tuple
  counts the percentage of optimal periods, and the average, median, and maximum
  slowdown ratios. Each pair indicates the average number of cores used
  according to their type. Presented results are averaged accross multiple
  runs.}
\label{tab:stats}
\resizebox{0.9\linewidth}{!}{
\begin{tabular}{c l c@{\hskip 0.05in}
                    r@{\hskip 0.05in}
                    r@{\hskip 0.05in}
                    r@{\hskip 0.05in}
                    r@{\hskip 0.05in}
                    c
                    c@{\hskip 0.05in}
                    r@{\hskip 0.05in}
                    r@{\hskip 0.05in}
                    c
                    c@{\hskip 0.05in}
                    r@{\hskip 0.05in}
                    r@{\hskip 0.05in}
                    r@{\hskip 0.05in}
                    r@{\hskip 0.05in}
                    c
                    c@{\hskip 0.05in}
                    r@{\hskip 0.05in}
                    r@{\hskip 0.05in}
                    c
                    c@{\hskip 0.05in}
                    r@{\hskip 0.05in}
                    r@{\hskip 0.05in}
                    r@{\hskip 0.05in}
                    r@{\hskip 0.05in}
                    c
                    c@{\hskip 0.05in}
                    r@{\hskip 0.05in}
                    r@{\hskip 0.05in}
                    c}
\toprule
& & \multicolumn{10}{c}{$SR = 0.2$} & \multicolumn{10}{c}{$SR = 0.5$} & \multicolumn{10}{c}{$SR = 0.8$} \\
\cmidrule(lr){3-12} \cmidrule(lr){13-22} \cmidrule(lr){23-32}
& & \multicolumn{6}{c}{Period Statistics} & \multicolumn{4}{c}{Core Usage} & \multicolumn{6}{c}{Period Statistics} & \multicolumn{4}{c}{Core Usage} & \multicolumn{6}{c}{Period Statistics} & \multicolumn{4}{c}{Core Usage} \\
\cmidrule(lr){3-8} \cmidrule(lr){9-12} \cmidrule(lr){13-18} \cmidrule(lr){19-22} \cmidrule(lr){23-28} \cmidrule(lr){29-32}
$R = (b,l)$ & Strategy & ( & \% opt, & avg, & med, & max & ) & ( & $b_{\text{used}}$, & $l_{\text{used}}$ & ) & ( & \% opt, & avg, & med, & max & ) & ( & $b_{\text{used}}$, & $l_{\text{used}}$ & ) & ( & \% opt, & avg, & med, & max & ) & ( & $b_{\text{used}}$, & $l_{\text{used}}$ & ) \\
\midrule
\multirow{5}{*}{\rotatebox[origin=c]{0}{$(16_{\mathcal{B}},4_{\mathcal{L}})$}}
    & \otac[L]  & ( &   0.0\%, & 9.0, & 8.9, & 13.8 & ) & ( &  0.0, &  4.0 & ) & ( &  0.0\%, & 9.3, & 9.2, & 14.8 & ) & ( &  0.0, &  4.0 & ) & ( &  0.0\%, & 10.5, & 10.3, & 17.9 & ) & ( &  0.0, &  4.0 & ) \\
    & \otac[B]  & ( &  88.7\%, & 1.0, & 1.0, &  1.3 & ) & ( & 14.1, &  0.0 & ) & ( & 82.7\%, & 1.0, & 1.0, &  1.3 & ) & ( & 14.3, &  0.0 & ) & ( & 69.9\%, &  1.0, &  1.0, &  1.4 & ) & ( & 14.4, &  0.0 & ) \\
    & \fertac   & ( &  99.2\%, & 1.0, & 1.0, &  1.1 & ) & ( & 12.4, &  3.9 & ) & ( & 95.8\%, & 1.0, & 1.0, &  1.2 & ) & ( & 12.8, &  3.9 & ) & ( & 84.3\%, &  1.0, &  1.0, &  1.3 & ) & ( & 13.3, &  3.8 & ) \\
    & \twocatac & ( & 100.0\%, & 1.0, & 1.0, &  1.0 & ) & ( & 11.7, &  3.3 & ) & ( & 99.6\%, & 1.0, & 1.0, &  1.1 & ) & ( & 12.0, &  3.4 & ) & ( & 93.0\%, &  1.0, &  1.0, &  1.1 & ) & ( & 12.9, &  3.3 & ) \\
    & \herad    & ( & 100.0\%, & 1.0, & 1.0, &  1.0 & ) & ( & 11.7, &  3.3 & ) & ( &100.0\%, & 1.0, & 1.0, &  1.0 & ) & ( & 11.9, &  3.5 & ) & ( &100.0\%, &  1.0, &  1.0, &  1.0 & ) & ( & 12.6, &  3.4 & ) \\
\addlinespace %\midrule
\multirow{5}{*}{\rotatebox[origin=c]{0}{$(10_{\mathcal{B}},10_{\mathcal{L}})$}}
    & \otac[L]  & ( &   0.0\%, & 4.1, & 4.1, &  5.6 & ) & ( &  0.0, &  9.5 & ) & ( &  0.0\%, & 4.3, & 4.3, &  5.8 & ) & ( &  0.0, &  9.7 & ) & ( &  0.0\%, &  4.3, &  4.4, &  5.8 & ) & ( &  0.0, &  9.8 & ) \\
    & \otac[B]  & ( &   1.7\%, & 1.3, & 1.3, &  1.7 & ) & ( &  9.9, &  0.0 & ) & ( &  1.4\%, & 1.3, & 1.3, &  1.8 & ) & ( &  9.9, &  0.0 & ) & ( &  1.6\%, &  1.4, &  1.4, &  1.9 & ) & ( &  9.9, &  0.0 & ) \\
    & \fertac   & ( &  80.3\%, & 1.0, & 1.0, &  1.2 & ) & ( &  9.4, &  8.8 & ) & ( & 51.2\%, & 1.0, & 1.0, &  1.4 & ) & ( &  9.4, &  9.8 & ) & ( & 42.2\%, &  1.0, &  1.0, &  1.3 & ) & ( &  9.5, &  9.8 & ) \\
    & \twocatac & ( &  98.8\%, & 1.0, & 1.0, &  1.0 & ) & ( &  9.3, &  7.9 & ) & ( & 89.1\%, & 1.0, & 1.0, &  1.2 & ) & ( &  9.1, &  9.2 & ) & ( & 61.7\%, &  1.0, &  1.0, &  1.2 & ) & ( &  9.3, &  9.3 & ) \\
    & \herad    & ( & 100.0\%, & 1.0, & 1.0, &  1.0 & ) & ( &  9.3, &  7.8 & ) & ( &100.0\%, & 1.0, & 1.0, &  1.0 & ) & ( &  9.0, &  9.2 & ) & ( &100.0\%, &  1.0, &  1.0, &  1.0 & ) & ( &  9.1, &  9.4 & ) \\
\addlinespace %\midrule
\multirow{5}{*}{\rotatebox[origin=c]{0}{$(4_{\mathcal{B}},16_{\mathcal{L}})$}}
    & \otac[L]  & ( &   0.0\%, & 2.2, & 2.1, &  4.7 & ) & ( &  0.0, & 10.9 & ) & ( &  0.0\%, & 2.5, & 2.4, &  4.7 & ) & ( &  0.0, & 11.9 & ) & ( &  0.0\%, &  2.5, &  2.3, &  4.9 & ) & ( &  0.0, & 13.2 & ) \\
    & \otac[B]  & ( &   0.0\%, & 1.6, & 1.5, &  2.6 & ) & ( &  4.0, &  0.0 & ) & ( &  0.0\%, & 2.0, & 2.0, &  2.8 & ) & ( &  4.0, &  0.0 & ) & ( &  0.0\%, &  2.4, &  2.4, &  3.1 & ) & ( &  4.0, &  0.0 & ) \\
    & \fertac   & ( &  99.0\%, & 1.0, & 1.0, &  1.0 & ) & ( &  3.9, &  9.2 & ) & ( & 61.4\%, & 1.0, & 1.0, &  1.3 & ) & ( &  3.9, & 14.0 & ) & ( & 13.0\%, &  1.0, &  1.0, &  1.3 & ) & ( &  3.9, & 15.9 & ) \\
    & \twocatac & ( & 100.0\%, & 1.0, & 1.0, &  1.0 & ) & ( &  3.9, &  7.8 & ) & ( & 91.7\%, & 1.0, & 1.0, &  1.1 & ) & ( &  3.9, & 13.4 & ) & ( & 41.1\%, &  1.0, &  1.0, &  1.2 & ) & ( &  3.9, & 15.8 & ) \\
    & \herad    & ( & 100.0\%, & 1.0, & 1.0, &  1.0 & ) & ( &  3.9, &  7.8 & ) & ( &100.0\%, & 1.0, & 1.0, &  1.0 & ) & ( &  3.9, & 13.3 & ) & ( &100.0\%, &  1.0, &  1.0, &  1.0 & ) & ( &  3.9, & 15.8 & ) \\
\bottomrule
\end{tabular}
}
\end{table*}

Although these results are related to our exact simulation parameters, their
general trends are the same for longer task chains or different numbers of
resources. Additional experiments (not covered here for the sake of space) have
revealed that non-optimal strategies tend to perform worse when more tasks have
to be scheduled (more decisions to make), but better when more resources are
available (easier to have enough resources for the slowest stage).

\subsection{Core Usage}

Table~\ref{tab:stats} also provides the average number of big and little cores
used by each scheduling strategy for different resources available and SRs. As
our secondary objective is to use as many little cores as necessary to reduce
power consumption (Section~\ref{sec:def}), using more little cores and less big
ones is desirable. In general, strategies use more cores when more tasks are
replicable (right col., $SR=0.8$) to reduce the period.

We can see that \twocatac tends to use almost the same number of resources as
\herad. It uses at most $0.3$ more cores than the minimal, sometimes using more
big and less little cores. \fertac, in its part, tends to use more of both
resources in its solution. By greedily trying to use little cores in earlier
stages, it ends up missing opportunities to make better use of these cores later
in the pipeline. Nonetheless, even in its worst average results, \fertac
requires~$1.41$ little cores ($R=(4_{\mathcal{B}},16_{\mathcal{L}}), SR=0.2$)
or~$1.36$ cores in total ($R=(16_{\mathcal{B}},4_{\mathcal{L}}), SR=0.5$) more
than \herad.

\begin{figure}[!b]
  \centering
  \begin{subfigure}{0.47\columnwidth}
      \includegraphics[width=\columnwidth]%
      {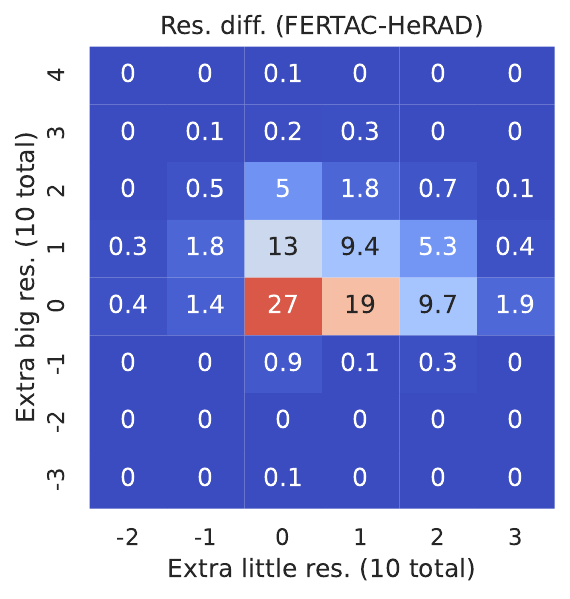}
      \caption{All results.}
      \label{fig:heat_all}
  \end{subfigure}
  % \hspace{1.0cm}
  \begin{subfigure}{0.47\columnwidth}
      \includegraphics[width=\columnwidth]%
      {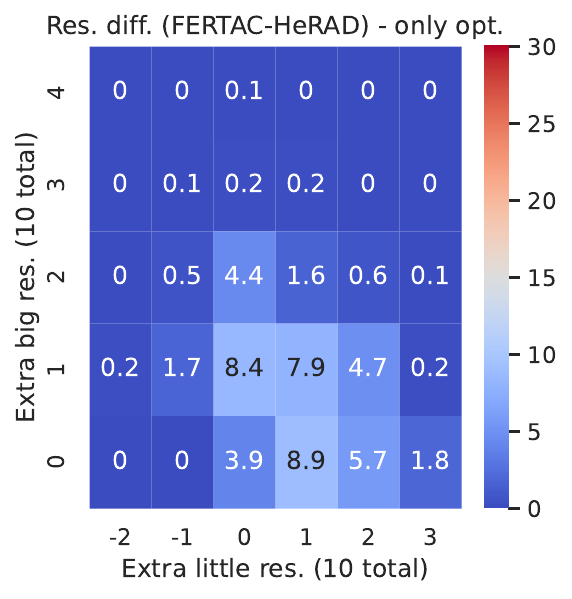}
      \caption{Only optimal periods.}
      \label{fig:heat_opt}
  \end{subfigure}
  \caption{Heatmaps with the differences (percentages) in resources used between
    \fertac and \herad for $R=(10_{\mathcal{B}},10_{\mathcal{L}})$ and
    $SR=0.5$.}
  \label{fig:heatmap}
\end{figure}

\begin{figure*}[htp]
  \centering
  \begin{subfigure}{0.475\columnwidth}
    \includegraphics[width=\columnwidth]%
    {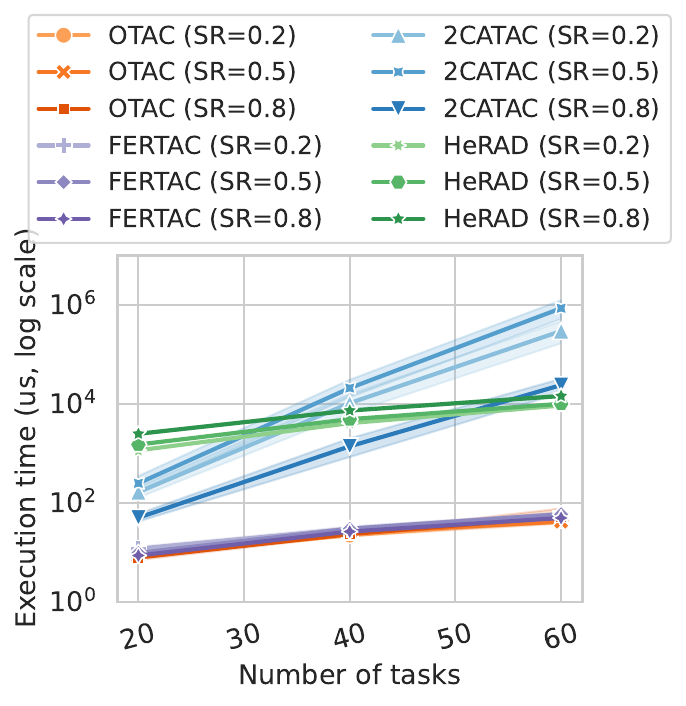}
    \caption{$R=(20_{\mathcal{B}},20_{\mathcal{L}})$.}
    \label{fig:20_res_split}
  \end{subfigure}
  \hspace{0.15cm}
  \begin{subfigure}{0.475\columnwidth}
    \includegraphics[width=\columnwidth]%
    {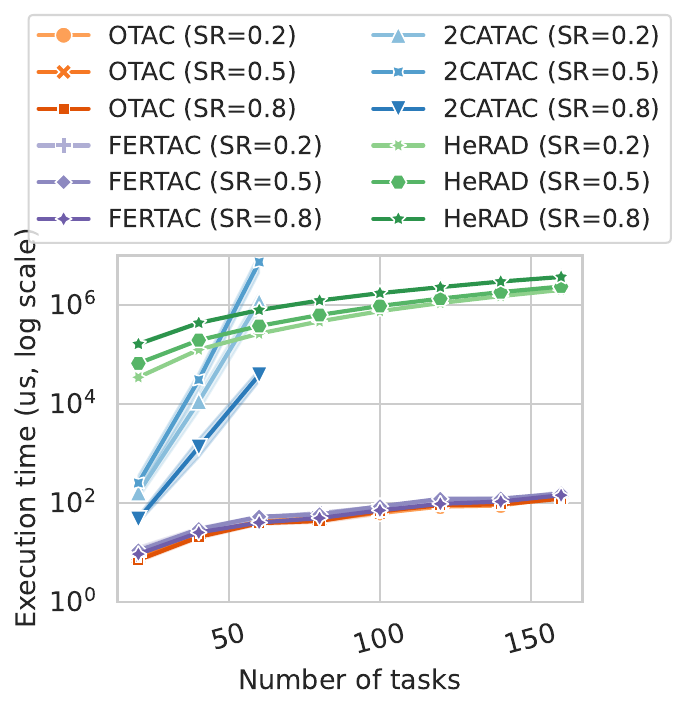}
    \caption{$R=(100_{\mathcal{B}},100_{\mathcal{L}})$.}
    \label{fig:100_res_split}
  \end{subfigure}
  %\label{fig:resources_split}
  \hspace{0.15cm}
  \begin{subfigure}{0.475\columnwidth}
    \includegraphics[width=\columnwidth]%
    {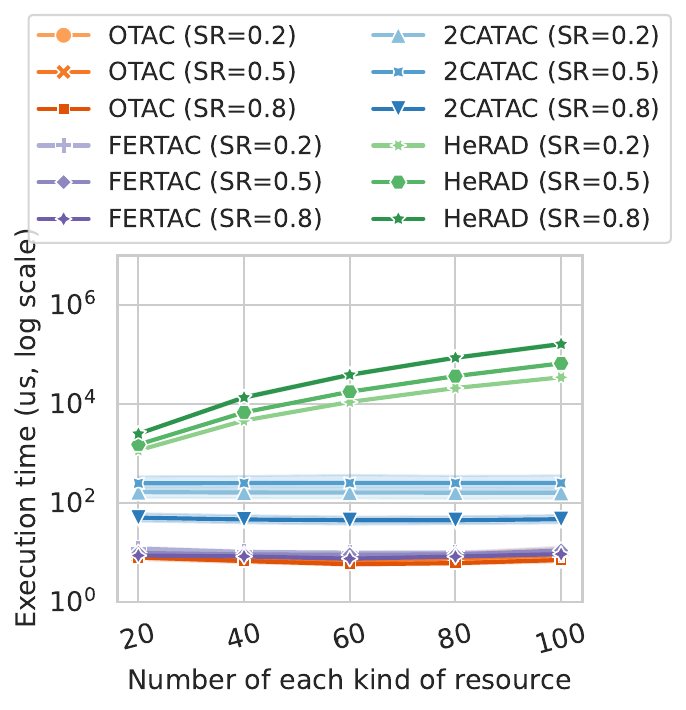}
    \caption{$|\tasks{}|=20$.}
    \label{fig:20_tasks}
  \end{subfigure}
  \hspace{0.15cm}
  \begin{subfigure}{0.475\columnwidth}
    \includegraphics[width=\columnwidth]%
    {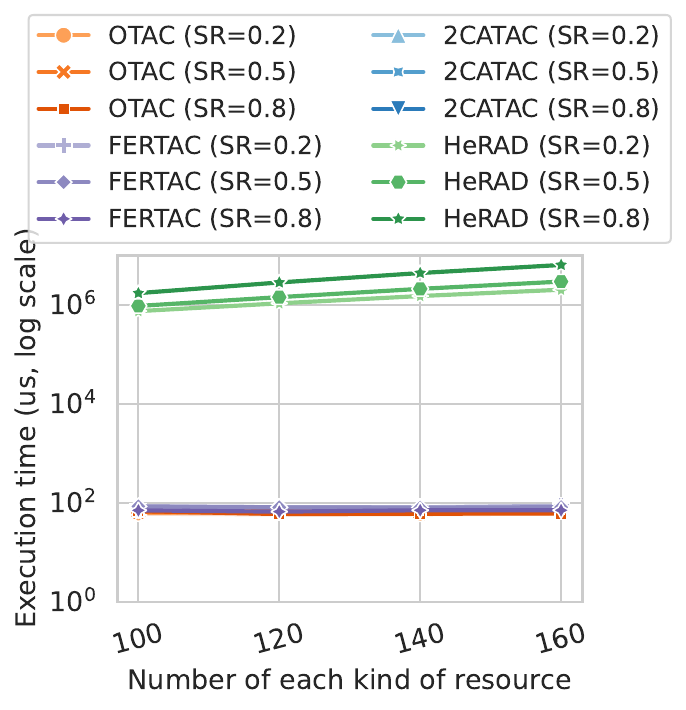}
    \caption{$|\tasks{}|=100$.}
    \label{fig:100_tasks}
  \end{subfigure}
  \caption{Average strategy times ($\mu$s, log. scale) with fixed numbers of
    resources (Subfigures~\ref{fig:20_res_split} and~\ref{fig:100_res_split}) or
    tasks (Subfigures~\ref{fig:20_tasks} and~\ref{fig:100_tasks}).}
  \label{fig:ress_tasks}
\end{figure*}

Figure~\ref{fig:heatmap} explores in more detail the differences between \fertac
and \herad for one scenario where \fertac achieved the minimum period $51.2\%$
of the times. Each cell in the heatmaps represent the percentage of times that
\fertac uses more, less, or the same number of big and little cores than
\herad. When considering all results (Figure~\ref{fig:heat_all}), \fertac uses
at most 1 or 2 extra cores $59\%$ and $83.1\%$ of the times, respectively. When
considering only the results where \fertac achieves minimal periods
(Figure~\ref{fig:heat_opt}), the situations where at most 1 or 2 extra cores
were necessary change to $21.2\%$ and $39.2\%$ of the times. These differences
may be justified given the difference in computational complexity of the
strategies, as will be seen next.

\subsection{Strategies Execution Times}

Figure~\ref{fig:ress_tasks} shows the execution times of the different
strategies in $\mu s$ for $R=(20_{\mathcal{B}},20_{\mathcal{L}})$
(Figure~\ref{fig:20_res_split}) and $R=(100_{\mathcal{B}},100_{\mathcal{L}})$
(Figure~\ref{fig:100_res_split}). Each point represents the average of 50 runs.
Each line represents a strategy computing schedules for task chains with
different stateless ratios. The lower the time, the better.

We can notice that \fertac displays the same behavior as \otac. Having the
lowest computational complexity among proposed strategies, its execution times
are in the order of $10$ to $100~\mu s$ and they grow proportionally to the
number of tasks.

\twocatac has an exponential complexity in the number of tasks, so its results
are limited to up to 60~tasks. Besides its rapid growth in execution time,
\twocatac shows distinct execution times depending on how many replicable tasks
there are. Its execution times increase when $SR$ goes from $0.2$ to~$0.5$, but
then it decreases by almost two orders of magnitude when $SR=0.8$. \twocatac is
able to pack many tasks together in longer pipeline stages when they are
replicable, leading to shorter recursions and fewer comparisons. Nonetheless,
its exponential behavior limits its usage to short task chains.

\herad's execution times grow with the square of the number of tasks, and they
already start in the order of ms in the tested scenarios. Its averages go from
$2.5$~ms to $14.4$~ms from 20 to 60 tasks ($R=(20_{\mathcal{B}},
20_{\mathcal{L}}), SR=0.8$), and from $78.5$~ms to $3656$~ms ($46.6\times$)
from 20 to 160 tasks ($R=(100_{\mathcal{B}},100_{\mathcal{L}}), SR=0.8$). Its
execution times are smaller when fewer tasks are replicable due to an
optimization in RecomputeCell (see Section~\ref{sec:dp}).

Figure~\ref{fig:ress_tasks} reflects the effects of increasing the number of
resources. It shows that the greedy strategies stay mostly unaffected, while
\herad's execution times grow. For instance, its execution times go from
$1.72~s$ to $6.38~s$ when going from $R=(100_{\mathcal{B}},100_{\mathcal{L}})$
to $R=(160_{\mathcal{B}},160_{\mathcal{L}})$ ($SR=0.8$ in
Figure~\ref{fig:100_tasks}) --- a $3.7\times$ increase in time for a $1.6\times$
increase in resources. Although these times are not prohibitive when
precomputing a schedule for contemporary task chains and processors, \herad
could be more difficult to use in bigger scenarios or real-time. We see how its
hypothetically optimal schedules behave when applied in a real scenario next.

\section{Runtime System}\label{sec:runtime}

\begin{figure*}[t]
  \centering
  \includegraphics[width=1\linewidth]{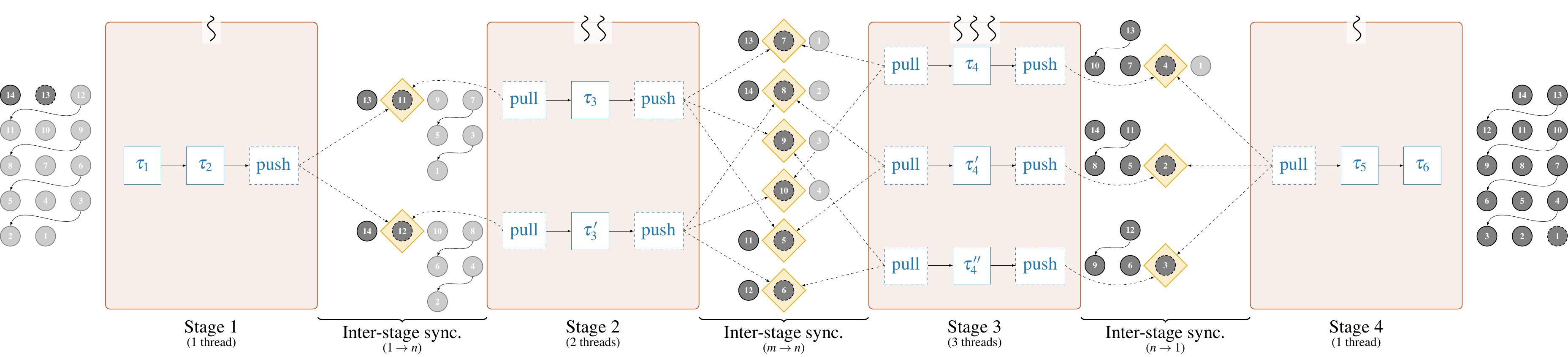}
  \caption{Example of the \spu inter-stage producer–consumer algorithm in a
    4-stage pipeline with 7 threads and 6 tasks, decomposed as follow:
    $(\tau_{1},\tau_{2})_{r = 1},(\tau_{3})_{r = 2},(\tau_{4})_{r = 3},
    (\tau_{5},\tau_{6})_{r = 1}$. 14 streams are considered. Black circles
    indicate streams that have not yet been processed, while light gray circles
    represent processed streams. Yellow diamonds denote shared buffers, each
    capable of storing a single stream. \emph{push} and \emph{pull} tasks (shown
    as rectangles with blue dashed borders) are automatically inserted by the
    runtime to handle synchronization. Dashed arrows highlight the buffers
    associated with these push and pull tasks.}
  \label{fig:synchros}
\end{figure*}

Previous section presented analytical results where communication and
synchronization overheads were neglected. In this section, we outline the main
aspects of our runtime implementation, which is provided as part of the \spu
open-source project~\cite{cassagne2023spu}.

\subsection{Features}

Thread synchronization relies on the portable C++11 thread library, which, on
Linux platforms, is implemented on top of the POSIX threads library
(\verb|pthread|). A thread pool is responsible for creating the threads at
program startup, and these threads are subsequently reused throughout execution.
This design reduces the overhead associated with repeatedly creating, pausing,
and resuming pipelines. The thread pool mechanism was first introduced in \spu
v1.2.1.

Then, threads are synchronized to exchange data buffers between pipeline
stages. Depending on the replication level, one or multiple threads may be
assigned to a given stage, and synchronization must ensure that the ordering of
input streams is preserved at the output. To achieve this, a producer–consumer
algorithm is employed, designed to minimize the reliance on critical sections in
the code. In particular, threads belonging to the same stage do not require
mutual synchronization, since they exclusively perform write (\emph{push}) or
read (\emph{pull}) operations on dedicated buffers. Consequently, the number of
buffers allocated is at least equal to the number of threads assigned to the
stage.

Two scenarios must be considered for inter-stage synchronization. The simpler
and most common case corresponds to a $1 \rightarrow n$ or $n \rightarrow 1$
configuration, where either a single thread in the first stage communicates with
$n$ threads in the next stage, or conversely, $n$ threads communicate with a
single thread. In such situations, exactly $n$ buffers are required. The
single-threaded stage either feeds or consumes these $n$ buffers in a
round-robin manner. Each buffer can store one or more streams and is associated
with two index variables, \verb|first| and \verb|last|. The \verb|first| index
points to the oldest stream available for consumption, while the \verb|last|
index refers to the most recently inserted stream. In addition, a \verb|counter|
variable tracks the number of available streams (i.e., streams ready for
consumption) within the buffer. Streams are stored in a circular fashion, and a
zero-copy mechanism based on pointer swapping is employed (see the original \spu
article for further details~\cite{cassagne2023spu}). The two indices and the
counter are safeguarded through atomic instructions, implemented in C++ with the
\verb|std::atomic<uint32_t>| type. On x86 architectures, these operations rely
on the hardware Fetch-and-Add (FAA) \verb|LOCK XADD| instruction, while on ARM
architectures they are supported from ISA version ARMv8.1 onward through the
\verb|LDADD| instruction.

The second case, which generalizes the $1 \rightarrow n$ and $n \rightarrow 1$
scenarios, arises when both consecutive stages contain multiple threads, denoted
as an $m \rightarrow n$ configuration. In this setting, $\max(m,n)$ buffers are
often insufficient to prevent additional synchronization among threads within
the same stage. A straightforward but naive solution is to allocate $m \times n$
buffers, such that each thread writes into a buffer in a round-robin manner with
a stride of $m$ or $n$. However, this approach is clearly suboptimal. In the
simplest case where $m = n$, allocating $m^2$ buffers is unnecessary, since $m$
buffers are sufficient. More generally, the optimal number of buffers
corresponds to the least common multiple (LCM) of $m$ and $n$, which
significantly reduces memory footprint compared to the $m \times n$ solution.
The implementation of this $m \rightarrow n$ producer–consumer algorithm
represents a contribution of this work and it is illustrated in
Figure~\ref{fig:synchros}. This feature was specifically introduced to provide
full support for the \fertac, \twocatac, and \herad algorithms, which may
generate solutions where consecutive stages are replicated and it has been
integrated into \spu v1.6.0.

Both passive and active waiting modes are supported; however, in this work, only
the passive mode is evaluated. It is implemented using the \verb|wait| and
\verb|notify| routines, relying on operating system timers and interrupts. The
passive mode ensures that a waiting thread is not executed on the CPU, thereby
avoiding unnecessary energy consumption--a factor of particular relevance for
this study. In the following real-world experiments, the pipeline buffer
capacity is set to a single stream. And, one stream corresponds to a set of
frames.

\subsection{From Scheduling to Runtime}

The proposed schedulers, such as \fertac, \twocatac, and \herad, are implemented
in a separate open-source project: AMP-scheduling~\cite{amp-scheduling_v2.0}.
The workflow proceeds as follows:
\begin{inparaenum}[(i)]
  \item the system is executed in \spu in sequential mode to collect profiling
    information on the different types of cores,
  \item then, these profiling results are provided as input to AMP-scheduling,
    which generates a schedule in JSON format,
  \item finally, \spu reads this JSON file and applies the computed scheduling
    at runtime.
\end{inparaenum}

Since the proposed scheduling strategies are static, using a JSON file
interface provides considerable flexibility, facilitating the integration of
future scheduling strategies and enabling easy adjustments to existing
schedules. Furthermore, the JSON interface allows tuning of runtime parameters,
such as passive or active waiting, synchronization buffer sizes, and thread
pinning strategies. In \spu v1.8.0, a specific scheduler, reading information
from a JSON file and generating a pipeline object, has been added.

\begin{listing}[htp]
\begin{minted}{json}
{
  "platform": "an heterogeneous CPU [3b,3l]",
  "resources": {
    "p-core": {
      "node-list": ["core0-2"], "cluster-size": 1, "smt": 2
    },
    "e-core": {
      "node-list": ["core3-5"], "cluster-size": 3, "smt": 1
    }
  },
  "scheduler-name": "FERTAC (4-stage solution)",
  "schedule": [ {
        "tasks": 1, "threads": 1, "core-type": "e-core",
        "pinning-policy": "packed", "sync-buff-size": 1,
        "sync-waiting-type": "active"
      },{
        "tasks": 1, "threads": 1, "core-type": "e-core",
        "pinning-policy": "packed", "sync-buff-size": 8,
        "sync-waiting-type": "passive"
      },{
        "tasks": 1, "threads": 1, "core-type": "p-core",
        "pinning-policy": "packed", "sync-buff-size": 1
      },{
        "tasks": 2, "threads": 1, "core-type": "e-core",
        "pinning-policy": "packed"
      }
    ]
}
\end{minted}
\caption{\fertac JSON schedule from Fig.~\ref{fig:fertac}, $R = (3_\mathcal{B},
  3_\mathcal{L})$.}
\label{lst:json}
\vspace{-0.3cm}
\end{listing}

Listing~\ref{lst:json} provides an example of a JSON file for the \fertac
solution illustrated in Figure~\ref{fig:fertac}. Lines 2–10 describe the
available resources: three big cores (or \verb|p-cores|) mapped to cores 0, 1,
and 2 ($C[(0),(1),(2)]_\mathcal{B}$), and three little cores (or \verb|e-cores|)
mapped to cores 3, 4, and 5 ($C[(3,4,5)]_\mathcal{L}$). In this paper, the
discovery of the number and types of cores relies on the \emph{hwloc}
library~\cite{broquedis2010hwloc} (via the \verb|hwloc-ls| executable).
Consequently, core identifiers are expressed using hwloc logical numbering.
Then, lines 11–27 define the schedule as an array of JSON objects. Each object
specifies a pipeline stage with: $n$ the number of tasks (\verb|tasks| field),
$r$ the number of replications (\verb|threads| field), and, $v$ the assignment
to a resource type (\verb|core-type| field). Additional optional fields, such as
\verb|pinning-policy|, \verb|sync-buff-size|, and \verb|sync-waiting-type|,
allow tuning of runtime parameters.

\section{Real-world SDR Experiments}\label{sec:sdr}

In this section, we employ our three strategies and \otac~\cite{orhan2025b} to
schedule an implementation of the DVB-S2 digital communication
standard~\cite{dvbs2} on \spu~\cite{cassagne2023spu} and on four heterogeneous
multicore processors. Then, we evaluate their throughput, energy consumption,
and the impact of thread pinning policies on performance. We provide more
details about our experimental environments and results in the next
subsections. Source code, result files, and scripts are freely available
online~\cite{artifact_v1.0}.

\subsection{Experimental Environment}

\begin{figure*}[htp!]
  \centering
  \includegraphics[width=1\linewidth]{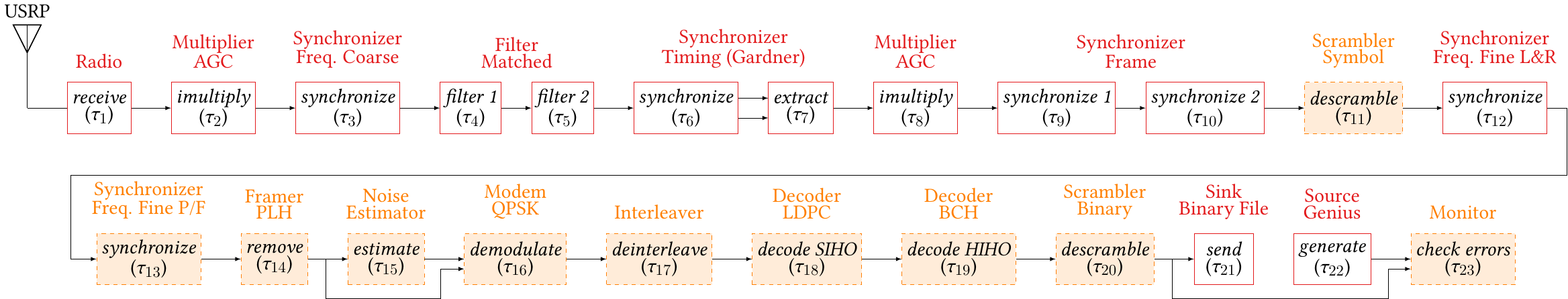}
  \caption{DVB-S2 receiver chain with data dependencies. Full, transparent
    (resp. dashed, shaded) boxes represent sequential (resp. replicable) tasks.}
  \label{fig:dvbs2_rx}
\end{figure*}

\begin{table*}[!h]
  \centering
  \caption{DVB-S2 receiver's average task latency on the evaluated platforms.
    The two slowest sequential and replicable tasks are highlighted in
    \colorbox{Paired-5!15}{red} and \colorbox{Paired-7!15}{orange},
    respectively. \spu streams contain multiple frames (or \emph{fra}),
    depending on the target platform and on the number of elements supported by
    the SIMD registers.}

  \label{tab:sdr_dvbs2_tasks_thr_lat}
  {\resizebox{0.8\linewidth}{!}{
  {%\small
  \begin{tabular}{l l c r r r r r r r r r r}
    \toprule
    & & & \multicolumn{8}{c}{Average Latency ($\mu$s)} \\
    \cmidrule(lr){4-11}
    \multicolumn{3}{c}{Task} & \multicolumn{2}{c}{Orange Pi 5+ (4 fra)} & \multicolumn{2}{c}{Mac Studio (4 fra)} & \multicolumn{2}{c}{AI370 (16 fra)} & \multicolumn{2}{c}{X7 Ti (8 fra)} \\
    \cmidrule(lr){1-3} \cmidrule(lr){4-5} \cmidrule(lr){6-7} \cmidrule(lr){8-9} \cmidrule(lr){10-11}
    Id & Name & Rep. & $\mathcal{B}$ & $\mathcal{L}$ & $\mathcal{B}$ & $\mathcal{L}$ & $\mathcal{B}$ & $\mathcal{L}$ & $\mathcal{B}$ & $\mathcal{L}$ \\
    \midrule
                            $\tau_{ 1}$  &                  Radio --               \emph{receive} &  \xm & $  193.4$ & $  319.3$ & $  54.0$ & $  245.6$ & $  216.3$ & $  227.1$ & $  139.2$ & $  178.5$ \\
                            $\tau_{ 2}$  &         Multiplier AGC --             \emph{imultiply} &  \xm & $  192.1$ & $  646.3$ & $  75.0$ & $  149.3$ & $  276.5$ & $  338.2$ & $  135.0$ & $  318.3$ \\
                            $\tau_{ 3}$  &     Sync. Freq. Coarse --           \emph{synchronize} &  \xm & $  476.1$ & $ 1052.0$ & $  96.2$ & $  496.1$ & $  156.7$ & $  206.9$ & $  111.6$ & $  428.2$ \\
                            $\tau_{ 4}$  &         Filter Matched --       \emph{filter (part 1)} &  \xm & $  861.1$ & $ 3202.9$ & $ 319.7$ & $  892.6$ & $  661.5$ & $  904.7$ & $  328.9$ & $  715.4$ \\
                            $\tau_{ 5}$  &         Filter Matched --       \emph{filter (part 2)} &  \xm & $  905.9$ & $ 3261.1$ & $ 316.9$ & $  880.0$ & $  711.3$ & $  935.7$ & $  326.2$ & $  747.0$ \\
    \rowcolor{Paired-5!15}  $\tau_{ 6}$  &           Sync. Timing --           \emph{synchronize} &  \xm & $ 1151.4$ & $ 3371.3$ & $ 950.9$ & $ 1470.8$ & $ 3262.0$ & $ 4307.5$ & $ 1342.5$ & $ 2397.8$ \\
                            $\tau_{ 7}$  &           Sync. Timing --               \emph{extract} &  \xm & $  137.9$ & $  346.7$ & $  50.7$ & $  105.6$ & $  135.4$ & $  179.2$ & $   59.5$ & $  142.8$ \\
                            $\tau_{ 8}$  &         Multiplier AGC --             \emph{imultiply} &  \xm & $   98.0$ & $  322.0$ & $  36.9$ & $   75.2$ & $  132.8$ & $  162.9$ & $   63.5$ & $  158.2$ \\
    \rowcolor{Paired-5!15}  $\tau_{ 9}$  &            Sync. Frame --  \emph{synchronize (part 1)} &  \xm & $  999.6$ & $ 3803.3$ & $ 361.9$ & $ 1051.3$ & $  758.9$ & $ 1037.2$ & $  364.6$ & $  849.7$ \\
                            $\tau_{10}$  &            Sync. Frame --  \emph{synchronize (part 2)} &  \xm & $  226.8$ & $  606.1$ & $  53.6$ & $  168.0$ & $  242.5$ & $  266.7$ & $   83.0$ & $  208.5$ \\
                            $\tau_{11}$  &       Scrambler Symbol --            \emph{descramble} &  \cm & $   80.8$ & $  178.6$ & $  16.2$ & $   60.9$ & $   72.6$ & $   76.6$ & $   24.6$ & $   68.7$ \\
                            $\tau_{12}$  &  Sync. Freq. Fine L\&R --           \emph{synchronize} &  \xm & $  239.3$ & $  535.9$ & $  50.1$ & $  246.2$ & $   85.8$ & $  111.1$ & $   54.2$ & $  203.5$ \\
                            $\tau_{13}$  &   Sync. Freq. Fine P/F --           \emph{synchronize} &  \cm & $  481.2$ & $ 1496.4$ & $  99.5$ & $  448.2$ & $  377.6$ & $  499.7$ & $  209.5$ & $  367.2$ \\
                            $\tau_{14}$  &             Framer PLH --                \emph{remove} &  \cm & $   75.2$ & $  214.3$ & $  25.3$ & $   64.9$ & $  312.6$ & $  316.3$ & $   55.0$ & $  100.1$ \\
                            $\tau_{15}$  &        Noise Estimator --              \emph{estimate} &  \cm & $   47.6$ & $  119.1$ & $  40.4$ & $   64.8$ & $   94.4$ & $  123.3$ & $   33.0$ & $   65.4$ \\
    \rowcolor{Paired-7!15}  $\tau_{16}$  &             Modem QPSK --            \emph{demodulate} &  \cm & $ 4246.1$ & $ 8853.7$ & $2260.9$ & $ 4831.2$ & $ 2491.1$ & $ 3281.4$ & $ 2122.3$ & $ 5750.1$ \\
                            $\tau_{17}$  &            Interleaver --          \emph{deinterleave} &  \cm & $   57.3$ & $  141.2$ & $  21.5$ & $   58.2$ & $   66.4$ & $   80.2$ & $   31.5$ & $   48.2$ \\
                            $\tau_{18}$  &           Decoder LDPC --           \emph{decode SIHO} &  \cm & $  698.9$ & $ 2822.3$ & $ 155.2$ & $  473.8$ & $  523.8$ & $  670.5$ & $  245.9$ & $ 1201.9$ \\
    \rowcolor{Paired-7!15}  $\tau_{19}$  &            Decoder BCH --           \emph{decode HIHO} &  \cm & $ 6342.1$ & $12779.8$ & $2653.3$ & $ 7240.2$ & $ 3959.0$ & $ 5220.4$ & $ 6038.0$ & $ 8121.8$ \\
                            $\tau_{20}$  &       Scrambler Binary --            \emph{descramble} &  \cm & $  601.3$ & $ 2170.5$ & $ 192.7$ & $  464.3$ & $  953.3$ & $ 1273.4$ & $  525.8$ & $  558.0$ \\
                            $\tau_{21}$  &       Sink Binary File --                  \emph{send} &  \xm & $   37.9$ & $  130.8$ & $   9.4$ & $   32.9$ & $   27.5$ & $   35.2$ & $   26.1$ & $   75.4$ \\
                            $\tau_{22}$  &          Source Genius --              \emph{generate} &  \xm & $   21.1$ & $   42.0$ & $   4.2$ & $   13.7$ & $   31.1$ & $   31.1$ & $   15.4$ & $   22.4$ \\
                            $\tau_{23}$  &                Monitor --          \emph{check errors} &  \cm & $   24.4$ & $  110.4$ & $   9.4$ & $   20.4$ & $   41.9$ & $   43.7$ & $    9.2$ & $   20.7$ \\
    \midrule
                                                                               \multicolumn{3}{r}{Total} & $18196.5$ & $46526.9$ & $7854.9$ & $19555.2$ & $15591.9$ & $20329.8$ & $12345.6$ & $22748.8$ \\
    \bottomrule
  \end{tabular}
  }
  }}
\end{table*}

Experiments were executed on four platforms, representative of the market
diversity at the time of writing:
\begin{enumerate}
  \item[(i)] an \textbf{Orange Pi 5+} SBC (Rockchip~RK3588 SoC with 4~ARM
    Cortex-A76 cores (big) @ 2.4~GHz  and 4~ARM Cortex-A55 cores (little) @
    1.8 GHz based on ARMv8.2-A ISA with 128-bit NEON SIMD, 16~GB LPDDR4, 256~GB
    eMMC) running Linux (Ubuntu~22.04.4~LTS, kernel~5.10.160, \textit{ondemand}
    CPUfreq governor ,g++~11.4.0).
    The OS scheduler is \textit{Completely Fair scheduler (CFS)} with the
    \textit{Capacity Aware Scheduler (CAS)} extension enabled, which helps the
    scheduler be aware of the heterogeneous architecture, but makes no choices
    to improve energy efficiency.
  \item[(ii)] an Apple \textbf{Mac Studio} 2022 (Apple Silicon M1 Ultra with
    16~Firestorm cores (big) @ 3.2~GHz and 4~Icestorm cores (little) @ 2~GHz
    based on the ARMv8.5-A ISA with 128-bit NEON SIMD, 64~GB LPDDR5 @ 6400~MT/s,
    2~TB SSD) running Linux (Fedora~40 Asahi~Remix kernel~6.11.8-400,
    \textit{schedutil} CPUfreq governor, g++~14.2.1).
    The OS scheduler is \textit{Earliest Eligible Virtual Deadline First
    (EEVDF)} with the \textit{Energy Aware Scheduling (EAS)} extension enabled,
    which include the heterogeneity awareness and  aims at improving energy
    efficiency-performance tradeoff.
  \item[(iii)] a Minisforum EliteMini \textbf{AI370} PC (AMD Ryzen AI~9~HX~370,
    4 Zen 5 cores (big) @ 2~GHz and 8 Zen 5c cores (little) @ 2~GHz based on the
    x86-64 ISA with 512-bit AVX-512 SIMD, 32~GB LPDDR5X @ 7500~MT/s, 1~TB SSD)
    running Linux (Ubuntu~24.04.2~LTS kernel~6.8.000, \textit{powersave} CPUfreq
    governor, g++~13.3.0).
    The OS scheduler is \EEVDF, with no extension enabled, because the
    architecture is not considered as heterogeneous by the kernel.
  \item[(iv)] a Minisforum AtomMan \textbf{X7 Ti} PC (Intel Ultra 9 185H,
    6~Redwood Cove p-cores (big) @ 2.3~GHz, 8~Crestmont e-cores (little) @
    1.8~GHz, and 2~Crestmont LPe-cores @ 1.0~GHz left unused based on the x86-64
    ISA with 256-bit AVX2 SIMD, 32~GB DDR5 @ 5600~MT/s, 1~TB NVMe SSD) running
    Linux (Ubuntu~24.10 kernel~6.11.0-13, \textit{powersave} CPUfreq governor,
    g++~14.2.0).
    The OS scheduler is \EEVDF, with the capacity detected, but the firmware
    gives same capacity for all cores, so no heterogeneous awareness.
\end{enumerate}

The reported frequencies are the base clocks. There is a boost technology on
AI370 and X7 Ti, enabling their processors to reach higher frequencies than the
ones reported in the previous paragraph, for the AI370 the difference between
the two type of cores is the maximum reachable frequency, 5.1~GHz for Zen~5
and 3.3~GHz Zen~5c, 2-way SMT is enabled for these cores and for Intel
Redwood~Cove cores, but it is not specifically targeted in this work.

All the platforms run \spu~v1.8.0~\cite{spu_v1.8} and the open source DVB-S2
transceiver v1.0.0~\cite{dvbs2_v1.0}. After profiling the DVB-S2 receiver on
all the four platforms\footnote{DVB-S2 receiver parameters: transmission phase,
running for 1~minute, inter-frame level $\in \{4_\text{NEON}, 8_\text{AVX2},
16_\text{AVX-512}\}$, $K = 14232$, $N = 16740$, $R = 8/9$, MODCOD~2, LDPC
horizontal layered NMS 10 iterations with early stop criterion, error-free SNR
zone.} (Table~\ref{tab:sdr_dvbs2_tasks_thr_lat}), schedules were computed using
all cores and half of them. Different thread pinning policies were used (more in
Section~\ref{subsec:pinning}). Each schedule was executed ten times for 1~minute
each and the achieved throughputs (in Mb/s) were obtained. This period
was chosen to measure the energy at steady point.

Energy measurements from the socket were made using the Dalek dedicated power
measurement platform~\cite{cassagne2025}. This dedicated hardware is able to get
1000 samples per second at the milliwatt resolution. The measurements are made
on the DC side (between the power supply and the mini-PCs). Thus, it excludes
the energy consumption of the power supply itself. In our tests, with a low
resolution external wattmeter, we observed that the power supply consumes 1 to
3~Watts to convert  220 V AC to DC. For Mac Studio, we were not able to use the
Dalek's measurement platform because the PSU is located inside the Mac case.
Thus, a low resolution external wattmeter (between the AC socket and the PSU
input) was used.

Energy measurements based on MSRs were also conducted on AI370 and X7 Ti
platforms. They are based on Intel RAPL \emph{pkg-0} domain which focus on the
whole SoC consumption. This domain excludes the energy consumption of the RAM
and the motherboard. The \verb|powerstat| tool is used and one sample every
2~seconds is taken in order to limit the probing effect.

\subsection{DVB-S2 Receiver Chain}

Figure~\ref{fig:dvbs2_rx} provides an overview of the DVB-S2 receiver. The role
of the different tasks is detailed in the \spu{} article~\cite{cassagne2023spu}
and is not repeated here. From a purely computational perspective, the studied
processing chain is composed of tasks with heterogeneous computational
signatures.

In Table~\ref{tab:sdr_dvbs2_tasks_thr_lat}, each task has been profiled on each
platform. The most time-consuming tasks are $\tau_{[4,5]}$, $\tau_{6}$,
$\tau_{9}$, $\tau_{16}$, $\tau_{18}$, $\tau_{19}$ and $\tau_{20}$.
$\tau_{[4,5]}$ corresponds to a FIR filter that has been manually split into
two stages to benefit from pipeline parallelism. Its time complexity is
$\mathcal{O}(b N/4)$, where $N$ is the coded frame size in bits ($N = 16740$ in
the studied configuration) and $b$ the filter length. The implementation heavily
relies on vectorized code through the SIMD \MIPP wrapper~\cite{cassagne2018} and
is compute-bound due to its high arithmetic intensity ($b = 81$). $\tau_{6}$ is
one of the least optimized tasks in the system: the current implementation
involves numerous indirect memory accesses and branch instructions. Its
complexity is $\mathcal{O}(N/4)$ and it is latency-bound. $\tau_{9}$ is mainly
composed by two FIR filters, vectorized as in $\tau_{[4,5]}$, but with smaller
$b$ values ($b_1 = 25$ and $b_2 = 64$); as a result, it tends to lie between
memory- and compute-bound on CPUs. $\tau_{16}$ is a vectorized QPSK demodulation
with $\mathcal{O}(N/4)$ complexity and is compute-bound. $\tau_{18}$ is an LDPC
channel decoder and represents an important computational cost. Its complexity
can be expressed as $\mathcal{O}(d_v \times i \times N)$, where $d_v$ is the
average variable-node degree (typically $3 \leq d_v \leq 4$ in DVB-S2
parity-check matrices) and $i$ the number of iterations ($i = 10$ in this work).
This task can be either latency- or memory-bound; note that the LDPC decoder is
heavily optimized and this is why its execution time is reduced in the profiled
system. The code implements an early stop criterion ($i \leq 10$) and it has
been carefully vectorized through \MIPP. $\tau_{19}$ is a BCH channel decoder
whose complexity is $\mathcal{O}(tN)$, where $t$ denotes the error-correction
capability ($t = 8$ in the following experiments); as for $\tau_{18}$, this task
exhibits a mix of latency- and memory-bound behavior. It is worth noting that
this implementation has not been specifically optimized. Finally, $\tau_{20}$ is
a descrambler performing a memory rotation (and a light transformation) with
$\mathcal{O}(K)$ complexity (with $K = 14232$ the number of information bits)
and is therefore memory-bound.

Note that in such applications, the term memory-bound is imprecise, as it
encompasses both cache and RAM accesses. The nature of these accesses depends on
the frame size and the replication level, so the actual cost of a task may
differ from the values reported in Table~\ref{tab:sdr_dvbs2_tasks_thr_lat} when
multiple cores are used.
Similarly, compute-bound here differs from classical dense algebra scenarios,
where CPU peak floating-point performance can be achieved; in our case, peak
performance arises from a combination of floating-point operations, integer
arithmetic, comparisons, and bit-manipulation instructions.

\subsection{Solutions}

\begin{table*}[h]
  \centering
  \caption{Specification of the configurations $R$ used with the real-world
    DVB-S2 receiver. Limiting pipeline stages (according to the simulation) are
    highlighted in \colorbox{Paired-7!15}{orange} if replicable,
    \colorbox{Paired-5!15}{red} otherwise (packed pinning policy).
    $\mathcal{E}$~/~fra is the energy consumed to process one frame of
    14232~information bits.}
  \label{tab:dvbs2}
  {\resizebox{1.0\textwidth}{!}{
  \begin{tabular}{l c r l l l r r r r r r r r r r}
  \toprule
  & \multicolumn{1}{c}{} & \multicolumn{7}{c}{Solution} & \multicolumn{4}{c}{Info. Throughput (Mb/s)} & \multicolumn{2}{c}{Socket Cons.} \\
  \cmidrule(lr){3-9} \cmidrule(lr){10-13} \cmidrule(lr){14-15}
  &             &    &          &                                                                                                          &                 &                  &                  &    Period &      &      &       &       & Power & $\mathcal{E}$ / fra. \\
  & $R = (b,l)$ & Id & Strategy & Pipeline decomposition where a stage is $(n^{\text{tasks}},r_{v \in \{\mathcal{L},\mathcal{B}\}})$ & $|\mathsf{s}|$ & $b_\text{used}$ & $l_\text{used}$ & ($\mu$s) & Sim. & Real & Diff. & Ratio &   (W) &                  (mJ) \\
  \midrule
  \multirow{11}{*}{\rotatebox[origin=c]{90}{Orange Pi 5+}}
  & \multirow{5}{*}{$(2_{\mathcal{B}},2_{\mathcal{L}})$} & $\mathcal{S}_{1}$ & \otac[L] & $\StaL{15}{1},\StaL[h]{8}{1}$ & $2$ & $0$ & $2$ & $27050.9$ & $2.1$ & $2.1$ & $+0.0$ & $+1$\% & $3.9$ & $26.9$ \\
  & & $\mathcal{S}_{2}$ & \otac[B] & $\StaB[h]{16}{1},\StaB{7}{1}$ & $2$ & $2$ & $0$ & $10413.3$ & $5.5$ & $5.4$ & $+0.1$ & $+1$\% & $5.8$ & $15.3$ \\
  & & $\mathcal{S}_{3}$ & \fertac & $\StaL{4}{1},\StaL{3}{1},\StaB[h]{11}{1},\StaB{5}{1}$ & $4$ & $2$ & $2$ & $7251.4$ & $7.9$ & $7.6$ & $+0.2$ & $+3$\% & $6.8$ & $12.7$ \\
  & & $\mathcal{S}_{4}$ & \twocatac & $\StaB{15}{1},\StaL{3}{2},\StaB[h]{5}{1}$ & $3$ & $2$ & $2$ & $7027.0$ & $8.1$ & $8.2$ & $-0.1$ & $-1$\% & $6.9$ & $12.0$ \\
  & & $\mathcal{S}_{5}$ & \herad & $\StaB{13}{1},\StaL{5}{2},\StaB[h]{5}{1}$ & $3$ & $2$ & $2$ & $7027.0$ & $8.1$ & $8.2$ & $-0.1$ & $-0$\% & $7.0$ & $12.2$ \\
  \addlinespace
  & \multirow{5}{*}{$(4_{\mathcal{B}},4_{\mathcal{L}})$} & $\mathcal{S}_{6}$ & \otac[L] & $\StaL{8}{1},\StaL{7}{1},\StaL{3}{1},\StaL[h]{5}{1}$ & $4$ & $0$ & $4$ & $15233.7$ & $3.7$ & $3.7$ & $+0.0$ & $+0$\% & $4.6$ & $17.5$ \\
  & & $\mathcal{S}_{7}$ & \otac[B] & $\StaB{12}{1},\StaB[h]{7}{2},\StaB{4}{1}$ & $3$ & $4$ & $0$ & $5974.4$ & $9.5$ & $9.4$ & $+0.1$ & $+0$\% & $7.9$ & $11.8$ \\
  & & $\mathcal{S}_{8}$ & \fertac & $\StaL[h]{4}{1},\StaL{1}{1},\StaL{3}{1},\StaL{4}{1},\StaB{8}{3},\StaB{3}{1}$ & $6$ & $4$ & $4$ & $5220.7$ & $10.9$ & $10.5$ & $+0.4$ & $+3$\% & $7.8$ & $10.6$ \\
  & & $\mathcal{S}_{9}$ & \twocatac & $\StaL{3}{1},\StaL{1}{1},\StaB{7}{1},\StaL{4}{1},\StaB[h]{4}{3},\StaL{4}{1}$ & $6$ & $4$ & $4$ & $3781.5$ & $15.1$ & $14.6$ & $+0.4$ & $+2$\% & $10.1$ & $9.9$ \\
  & & $\mathcal{S}_{10}$ & \herad & $\StaB{5}{1},\StaB{8}{1},\StaL{4}{3},\StaB[h]{2}{2},\StaL{4}{1}$ & $5$ & $4$ & $4$ & $3520.5$ & $16.2$ & $15.3$ & $+0.9$ & $+5$\% & $10.5$ & $9.8$ \\
  \midrule
  \multirow{11}{*}{\rotatebox[origin=c]{90}{Mac Studio}}
  & \multirow{5}{*}{$(8_{\mathcal{B}},2_{\mathcal{L}})$} & $\mathcal{S}_{11}$ & \otac[L] & $\StaL[h]{16}{1},\StaL{7}{1}$ & $2$ & $0$ & $2$ & $11251.4$ & $5.1$ & $5.0$ & $+0.1$ & $+1$\% & $16.6$ & $47.6$ \\
  & & $\mathcal{S}_{12}$ & \otac[B] & $\StaB{5}{1},\StaB{3}{1},\StaB{7}{1},\StaB[h]{4}{4},\StaB{4}{1}$ & $5$ & $8$ & $0$ & $1272.8$ & $44.7$ & $41.9$ & $+2.9$ & $+6$\% & $40.9$ & $13.9$ \\
  & & $\mathcal{S}_{13}$ & \fertac & $\StaL{3}{1},\StaL{1}{1},\StaB[h]{2}{1},\StaB{9}{1},\StaB{5}{5},\StaB{3}{1}$ & $6$ & $8$ & $2$ & $1267.9$ & $44.9$ & $42.0$ & $+2.9$ & $+6$\% & $39.9$ & $13.5$ \\
  & & $\mathcal{S}_{14}$ & \twocatac & $\StaB{5}{1},\StaB{2}{1},\StaB{8}{1},\StaB[h]{4}{5},\StaL{4}{1}$ & $5$ & $8$ & $1$ & $1018.2$ & $55.9$ & $51.7$ & $+4.3$ & $+8$\% & $45.0$ & $12.4$ \\
  & & $\mathcal{S}_{15}$ & \herad & $\StaB{5}{1},\StaB{1}{1},\StaB{9}{1},\StaB[h]{4}{5},\StaL{4}{1}$ & $5$ & $8$ & $1$ & $1018.2$ & $55.9$ & $52.7$ & $+3.2$ & $+6$\% & $45.2$ & $12.2$ \\
  \addlinespace
  & \multirow{5}{*}{$(16_{\mathcal{B}},4_{\mathcal{L}})$} & $\mathcal{S}_{16}$ & \otac[L] & $\StaL[h]{14}{1},\StaL{5}{2},\StaL{4}{1}$ & $3$ & $0$ & $4$ & $6355.3$ & $9.0$ & $8.7$ & $+0.3$ & $+3$\% & $17.9$ & $29.3$ \\
  & & $\mathcal{S}_{17}$ & \otac[B] & $\StaB{5}{1},\StaB[h]{1}{1},\StaB{9}{1},\StaB{5}{6},\StaB{3}{1}$ & $5$ & $10$ & $0$ & $951.0$ & $59.9$ & $55.9$ & $+4.0$ & $+7$\% & $49.3$ & $12.6$ \\
  & & $\mathcal{S}_{18}$ & \fertac & $\StaL{3}{1},\StaL{1}{1},\StaL{1}{1},\StaB[h]{1}{1},\StaL{2}{1},\StaB{7}{1},\StaB{5}{6},\StaB{3}{1}$ & $8$ & $9$ & $4$ & $951.0$ & $59.9$ & $55.7$ & $+4.2$ & $+7$\% & $45.7$ & $11.7$ \\
  & & $\mathcal{S}_{19}$ & \twocatac & $\StaL{3}{1},\StaL{1}{1},\StaL{1}{1},\StaB[h]{1}{1},\StaB{9}{1},\StaB{5}{6},\StaL{3}{1}$ & $7$ & $8$ & $4$ & $951.0$ & $59.9$ & $55.8$ & $+4.1$ & $+7$\% & $46.0$ & $11.7$ \\
  & & $\mathcal{S}_{20}$ & \herad & $\StaL{3}{1},\StaL{1}{1},\StaL{1}{1},\StaB[h]{1}{1},\StaB{6}{1},\StaB{7}{6},\StaL{4}{1}$ & $7$ & $8$ & $4$ & $951.0$ & $59.9$ & $55.9$ & $+4.0$ & $+7$\% & $45.2$ & $11.5$ \\
  \midrule
  \multirow{11}{*}{\rotatebox[origin=c]{90}{AI370}}
  & \multirow{5}{*}{$(2_{\mathcal{B}},4_{\mathcal{L}})$} & $\mathcal{S}_{21}$ & \otac[L] & $\StaL{5}{1},\StaL{7}{1},\StaL{6}{1},\StaL[h]{5}{1}$ & $4$ & $0$ & $4$ & $6604.0$ & $34.5$ & $33.8$ & $+0.7$ & $+2$\% & $20.0$ & $8.4$ \\
  & & $\mathcal{S}_{22}$ & \otac[B] & $\StaB{15}{1},\StaB[h]{8}{1}$ & $2$ & $2$ & $0$ & $8094.5$ & $28.1$ & $27.4$ & $+0.7$ & $+2$\% & $20.7$ & $10.7$ \\
  & & $\mathcal{S}_{23}$ & \fertac & $\StaL{5}{1},\StaB{3}{1},\StaL{7}{1},\StaL[h]{3}{1},\StaB{1}{1},\StaL{4}{1}$ & $6$ & $2$ & $4$ & $4032.2$ & $56.5$ & $56.6$ & $-0.1$ & $-0$\% & $33.9$ & $8.5$ \\
  & & $\mathcal{S}_{24}$ & \twocatac & $\StaL{5}{1},\StaB{3}{1},\StaL{7}{1},\StaL[h]{3}{1},\StaB{1}{1},\StaL{4}{1}$ & $6$ & $2$ & $4$ & $4032.2$ & $56.5$ & $56.6$ & $-0.1$ & $-0$\% & $33.5$ & $8.4$ \\
  & & $\mathcal{S}_{25}$ & \herad & $\StaL{5}{1},\StaB{1}{1},\StaL{9}{1},\StaL[h]{3}{1},\StaB{1}{1},\StaL{4}{1}$ & $6$ & $2$ & $4$ & $4032.2$ & $56.5$ & $56.6$ & $-0.1$ & $-0$\% & $34.5$ & $8.7$ \\
  \addlinespace
  & \multirow{5}{*}{$(4_{\mathcal{B}},8_{\mathcal{L}})$} & $\mathcal{S}_{26}$ & \otac[L] & $\StaL{5}{1},\StaL[h]{1}{1},\StaL{9}{1},\StaL{5}{3},\StaL{3}{1}$ & $5$ & $0$ & $7$ & $4307.5$ & $52.9$ & $52.8$ & $+0.1$ & $+0$\% & $28.0$ & $7.6$ \\
  & & $\mathcal{S}_{27}$ & \otac[B] & $\StaB{5}{1},\StaB{7}{1},\StaB{6}{1},\StaB[h]{5}{1}$ & $4$ & $4$ & $0$ & $5013.0$ & $45.4$ & $45.2$ & $+0.3$ & $+0$\% & $31.0$ & $9.8$ \\
  & & $\mathcal{S}_{28}$ & \fertac & $\StaL{5}{1},\StaB[h]{1}{1},\StaL{9}{1},\StaL{4}{3},\StaL{4}{1}$ & $5$ & $1$ & $6$ & $3262.1$ & $69.8$ & $57.6$ & $+12.2$ & $+21$\% & $33.0$ & $8.2$ \\
  & & $\mathcal{S}_{29}$ & \twocatac & $\StaL{5}{1},\StaB[h]{1}{1},\StaL{9}{1},\StaL{4}{3},\StaL{4}{1}$ & $5$ & $1$ & $6$ & $3262.1$ & $69.8$ & $57.7$ & $+12.1$ & $+20$\% & $33.3$ & $8.2$ \\
  & & $\mathcal{S}_{30}$ & \herad & $\StaL{5}{1},\StaB[h]{1}{1},\StaL{7}{1},\StaL{6}{3},\StaL{4}{1}$ & $5$ & $1$ & $6$ & $3262.1$ & $69.8$ & $60.3$ & $+9.6$ & $+15$\% & $34.0$ & $8.0$ \\
  \midrule
  \multirow{11}{*}{\rotatebox[origin=c]{90}{X7 Ti}}
  & \multirow{5}{*}{$(3_{\mathcal{B}},4_{\mathcal{L}})$} & $\mathcal{S}_{31}$ & \otac[L] & $\StaL{15}{1},\StaL[h]{4}{2},\StaL{4}{1}$ & $3$ & $0$ & $4$ & $7561.1$ & $15.1$ & $14.6$ & $+0.5$ & $+3$\% & $27.2$ & $26.6$ \\
  & & $\mathcal{S}_{32}$ & \otac[B] & $\StaB{18}{1},\StaB[h]{1}{1},\StaB{4}{1}$ & $3$ & $3$ & $0$ & $6038.1$ & $18.9$ & $18.1$ & $+0.7$ & $+4$\% & $49.4$ & $38.8$ \\
  & & $\mathcal{S}_{33}$ & \fertac & $\StaL{5}{1},\StaL{3}{1},\StaL{7}{1},\StaB[h]{4}{3},\StaL{4}{1}$ & $5$ & $3$ & $4$ & $2812.6$ & $40.5$ & $35.3$ & $+5.2$ & $+14$\% & $75.8$ & $30.5$ \\
  & & $\mathcal{S}_{34}$ & \twocatac & $\StaL{5}{1},\StaB{10}{1},\StaB{3}{1},\StaL[h]{1}{3},\StaB{4}{1}$ & $5$ & $3$ & $4$ & $2707.3$ & $42.1$ & $37.4$ & $+4.6$ & $+12$\% & $64.8$ & $24.6$ \\
  & & $\mathcal{S}_{35}$ & \herad & $\StaB{5}{1},\StaB{10}{1},\StaB{3}{1},\StaL[h]{1}{3},\StaL{4}{1}$ & $5$ & $3$ & $4$ & $2707.3$ & $42.1$ & $38.7$ & $+3.4$ & $+8$\% & $66.6$ & $24.5$ \\
  \addlinespace
  & \multirow{5}{*}{$(6_{\mathcal{B}},8_{\mathcal{L}})$} & $\mathcal{S}_{36}$ & \otac[L] & $\StaL{5}{1},\StaL{5}{1},\StaL{5}{1},\StaL[h]{4}{4},\StaL{4}{1}$ & $5$ & $0$ & $8$ & $3780.6$ & $30.1$ & $26.1$ & $+4.0$ & $+15$\% & $30.6$ & $16.7$ \\
  & & $\mathcal{S}_{37}$ & \otac[B] & $\StaB{8}{1},\StaB{7}{1},\StaB[h]{4}{3},\StaB{4}{1}$ & $4$ & $6$ & $0$ & $2812.6$ & $40.5$ & $37.6$ & $+2.9$ & $+7$\% & $69.3$ & $26.2$ \\
  & & $\mathcal{S}_{38}$ & \fertac & $\StaL{3}{1},\StaL{2}{1},\StaB{3}{1},\StaL{4}{1},\StaL{6}{5},\StaB[h]{1}{4},\StaB{4}{1}$ & $7$ & $6$ & $8$ & $1509.5$ & $75.4$ & $63.7$ & $+11.7$ & $+18$\% & $79.8$ & $17.8$ \\
  & & $\mathcal{S}_{39}$ & \twocatac & $\StaB{5}{1},\StaB[h]{1}{1},\StaB{9}{1},\StaB{3}{2},\StaL{2}{7},\StaL{3}{1}$ & $6$ & $5$ & $8$ & $1342.5$ & $84.8$ & $71.0$ & $+13.8$ & $+19$\% & $77.0$ & $15.4$ \\
  & & $\mathcal{S}_{40}$ & \herad & $\StaB{5}{1},\StaB[h]{1}{1},\StaB{6}{1},\StaB{4}{2},\StaL{3}{7},\StaL{4}{1}$ & $6$ & $5$ & $8$ & $1342.5$ & $84.8$ & $71.5$ & $+13.3$ & $+18$\% & $78.7$ & $15.7$ \\
  \bottomrule
  \end{tabular}
  }}
\end{table*}

Table~\ref{tab:dvbs2} summarizes the solutions found on the different platforms
using half and all of their cores. They were obtained using the task profiling
information listed in Table~\ref{tab:sdr_dvbs2_tasks_thr_lat}. As expected, task
latency is always higher on little cores. However, the latency ratio between
little and big cores varies according to task and platform. It underlines the
need to profile each task independently. It is worth mentioning that, depending
on the platform, the number of frames processed when a task is triggered may
vary. For instance, the Mac Studio processes 4 frames per task execution while
the AI370 processes 16 frames. This is related to the SIMD length and the
vectorization strategy used in the LDPC decoder~\cite{Cassagne2021}. In
Table~\ref{tab:dvbs2}, for each configuration and strategy, the pipeline
decomposition is detailed, including the number of little and big cores used,
the expected period, and its conversion to throughput metrics. Besides the
estimations, the real average Mb/s values are presented with their absolute and
relative differences to the expected values. Finally, the socket instant power
and energy per frame is given. Performance is also shown in
Figure~\ref{fig:dvbs2}.

\subsection{Achieved Throughput}

\begin{figure}[h]
  \centering
  \begin{subfigure}{1.0\columnwidth}
    \includegraphics[width=\columnwidth]{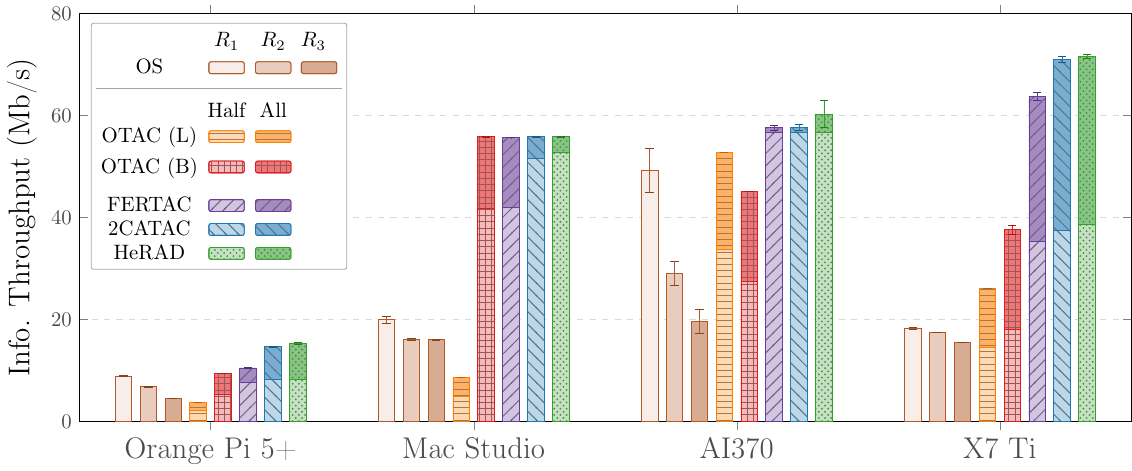}
    \caption{Throughput considering \emph{half} of the cores (light color) and
      \emph{all} of them (dark color). Min-max delimiters represent std dev.
      Higher is better.}
    \label{fig:dvbs2_thr}
  \end{subfigure}

  \vspace{0.4cm}
  \begin{subfigure}{1.0\columnwidth}
    \includegraphics[width=\columnwidth]{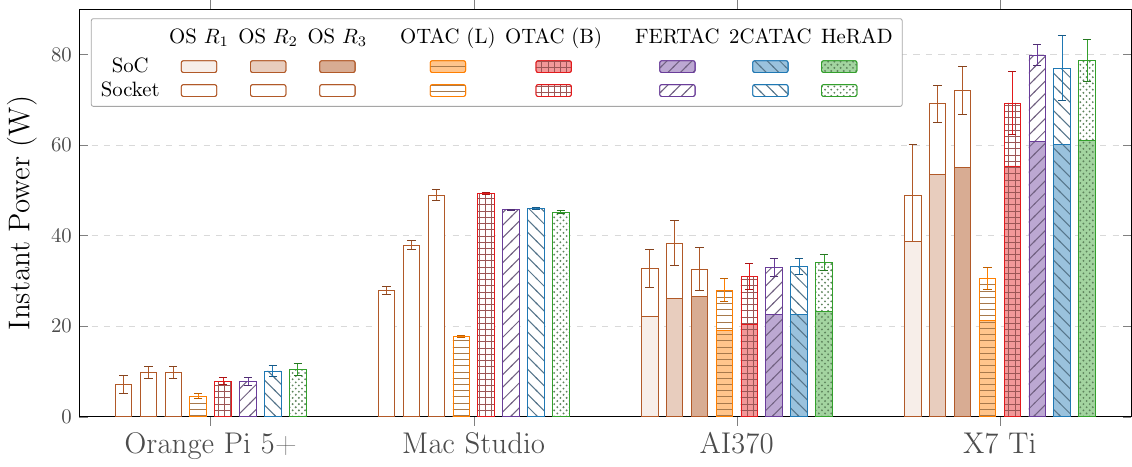}
    \caption{Instant power considering all the available cores. Min-max
      delimiters at the top of each bar represent the standard deviation. Lower
      is better.}
    \label{fig:dvbs2_power}
  \end{subfigure}

  \vspace{0.4cm}
  \begin{subfigure}{1.0\columnwidth}
    \includegraphics[width=\columnwidth]{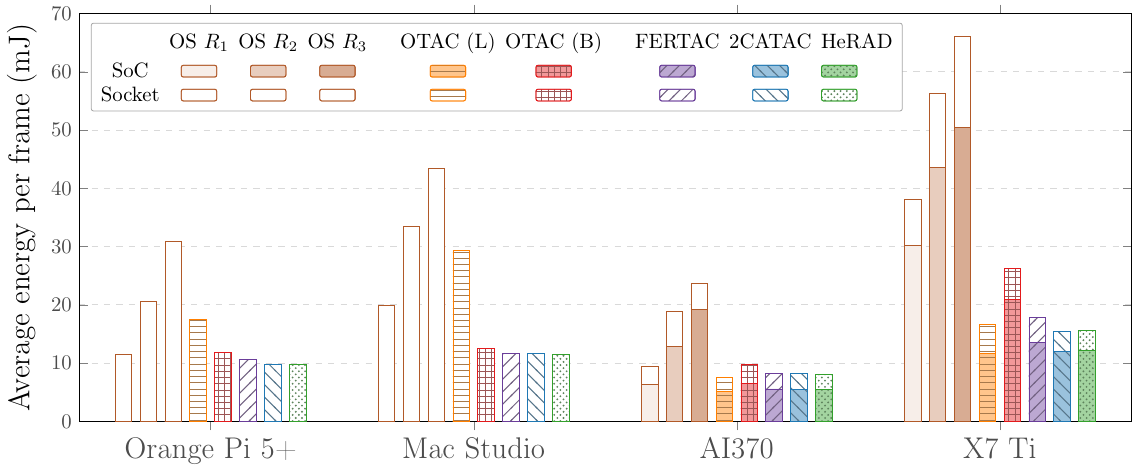}
    \caption{Energy per frame considering all the available cores. Lower
      is better.}
    \label{fig:dvbs2_energy}
  \end{subfigure}

  \caption{Achieved performance on the DVB-S2 receiver depending on the platform
    and scheduling strategy. The packed pinning policy is applied (except for
    the OS strategies, where the thread placement is left to it).}
  \label{fig:dvbs2}
\end{figure}

\fertac, \twocatac, and \herad propose different pipeline decompositions
(see Table~\ref{tab:dvbs2}), which result in different throughputs in practice.
The later are highlighted in Figure~\ref{fig:dvbs2_thr} and commented in the
following paragraphs.

An additional set of strategies named \emph{OS} is introduced in
Figure~\ref{fig:dvbs2_thr} (see plain brown bars). They rely on the operating
system scheduler: each task is allocated to one stage, thus leading to a fully
pipelined strategy with many threads that can be scheduled by the OS. $R_1$
means that the replication is disabled. This is the scheduling strategy used in
GNU~Radio~3~\cite{bloessl2019benchmarking}. Still, given that the limiting stage
in DVB-S2's receiver is replicable, we also consider strategies $R_2$ and $R_3$,
where \textit{every replicable stage} is replicated two or three times.
Strategies $R_1$, $R_2$, and $R_3$ generate 23, 33, and 43~threads,
respectively.

\textbf{Orange Pi 5+}.
For $R = (2_{\mathcal{B}},2_{\mathcal{L}})$, \twocatac and \herad achieve both
8.2~Mb/s with equivalent solutions. \fertac uses little cores for the two first
stages leading to a slower performance (7.6~Mb/s). \otac[B] lags behind with a
throughput of 5.4 Mb/s.

When all the cores are used ($R = (4_{\mathcal{B}},4_{\mathcal{L}})$), a
significant increase in throughput is seen for \fertac, \twocatac, and \herad.
This happens thanks to the replication of the slowest stage (see
$\mathcal{S}_{[8:10]}$). \herad's solution has fewer stages than \fertac's and
\twocatac's ones, thus achieving the best performance (reduced synchronizations
overhead). Its throughput is almost doubled when compared to the performance on
half of the cores (15.3 Mb/s).

\textbf{Mac Studio}.
When only half the cores are used ($R = (8_{\mathcal{B}},2_{\mathcal{L}})$),
\twocatac and \herad achieve the highest throughput by replicating $5\times$ the
stage with the two slowest tasks. \fertac ends up using both little cores for
the first stages, while using a single big core would have been better (see
$\mathcal{S}_{13}$ versus $\mathcal{S}_{[14:15]}$). As so, it lacks the extra
core by the end of the pipeline, leading to a lower throughput.

On the contrary, when all cores are available ($R = (16_{\mathcal{B}},
4_{\mathcal{L}})$), the strategies achieved similar throughput. With enough big
cores, performance gets limited by a sequential task, leaving many cores unused.
\otac[B] achieves similar throughput compared to \fertac, \twocatac, and \herad
but uses more big cores (10 vs 8 cores for \herad).

\textbf{AI370}.
The AMD Ryzen AI~9~HX~370 is the only tested CPU where using only the little
cores outperforms using only the big ones. Indeed, when $R = (2_{\mathcal{B}},
4_{\mathcal{L}})$, \otac[L] reaches 33.8 Mb/s while \otac[B] reaches only
27.4~Mb/s. This is mainly (but not only) due to the ratio between these two
types of cores. Zen~5 and Zen~5c are similar architectures that only differ from
their allowed maximum boost frequencies and L3 cache sizes. Using half of the
cores, \fertac, \twocatac, and \herad achieve roughly the same performance
(56.6~Mb/s).

When all the cores are used ($R = (4_{\mathcal{B}},8_{\mathcal{L}})$), we do not
see much improvement for \fertac and \twocatac (see $\mathcal{S}_{[28:29]}$).
\herad, meanwhile, is able to reach a higher throughput by better balancing the
cores between stages 3 and 4. In all cases, performance is limited by the
sequential stage 2.

Additionally, we see that this configuration shows high differences between
expected and obtained throughput results ($\geq 15\%$). This is mainly due to
biases in the model: as the profiling phase is run using a single core, its
clock frequency is allowed to be boosted to very high frequencies. On our
mini-PC, during profiling, we observed constant 4370~MHz and 3300~MHz
frequencies on the big and little cores, respectively. When we measured these
frequencies again while executing \herad ($\mathcal{S}_{30}$), we observed the
same frequency on the big core while the frequency of the 6 little cores ranged
from 1680 to 3300~MHz (up to 49\% lower). The little cores are clearly not able
to maintain the highest clock frequency when they are all used together, leading
to a significant difference between estimations and reality. However, one
positive side effect of maximizing the little cores usage is to relax the
pressure on the big core, thus allowing it to keep a stable and high frequency.

\textbf{X7 Ti}.
On this platform, there is a large gap between the schedules using half or all
cores. For $R = (3_{\mathcal{B}},4_{\mathcal{L}})$, \fertac, \twocatac, and
\herad use all cores to get to the point of being limited by the slowest
replicable tasks. Meanwhile, \otac[B] only gets to $47\%$ of \herad's
throughput, emphasizing the need of using both types of cores.

When $R = (6_{\mathcal{B}},8_{\mathcal{L}})$, all solutions from our scheduling
strategies ($\mathcal{S}_{[38:40]}$) have two consecutive replicated stages
using different types of cores. These required an extension to \spu to connect
replicated stages. This feature was unavailable before because, when using only
homogeneous resources, it is always better to merge consecutive replicated
stages~\cite{benoit2010complexity}. This enhancement has been released in
\spu~v1.6.0 (see Section~\ref{sec:runtime}).

Similarly to AI370, there are high differences between expected and obtained
throughput results ($\geq 18\%$). During the profiling phase, we observed a
frequency range of 4600-4800~MHz on the big core and 3600-3800~MHz on the little
one. Then, we measured these frequencies again while executing \herad
($\mathcal{S}_{40}$). We observed a frequency range of 3700-4600~MHz on the
5~big cores (between 23\% and 4\% lower) while the frequency of the 8 little
cores was stable around 3300~MHz (about 13\% lower). In other words, both types
of cores showed lower clock frequencies when using multiple cores in comparison
to the profiling phase.

\textbf{OS Strategies}.
Figure~\ref{fig:dvbs2_thr} shows that, except on AI370 with $R_1$, the OS
performances are disappointing. Increasing the replication factor ($R_2$ and
$R_3$) leads to even worse performance. These results clearly demonstrate the
interest of the proposed scheduling algorithms. Moreover, they highlight the
fact that the replication has to be combined with adapted scheduling strategies
to reveal its full potential. These conclusions can be tempered by considering
that generating as many stages as there are tasks leads to unnecessary
synchronization overhead. The performance of custom OS schedulers is studied
later in Subsection~\ref{subsec:os_scx}.

\subsection{Power and Energy Consumption}

\begin{figure*}[htp!]
  \centering
  % First row
  \begin{subfigure}[b]{0.225\linewidth}
    \centering
    \includegraphics[width=1\columnwidth]{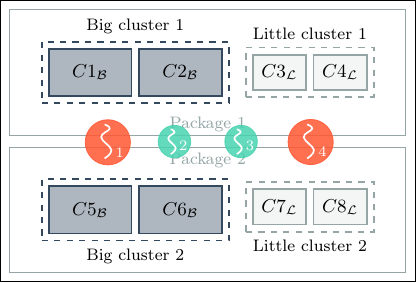}
    \caption{Loose policy.}
  \label{fig:pin_loose}
  \end{subfigure}
  \quad
  \begin{subfigure}[b]{0.225\linewidth}
    \centering
    \includegraphics[width=1\columnwidth]{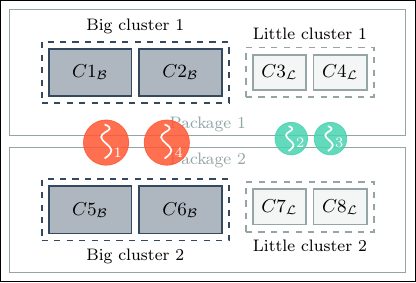}
    \caption{Guided policy.}
    \label{fig:pin_guided}
  \end{subfigure}
  \quad
  % Second row
  \begin{subfigure}[b]{0.225\linewidth}
    \centering
    \includegraphics[width=1\columnwidth]{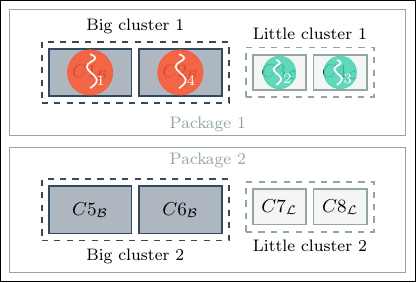}
    \caption{Packed policy.}
    \label{fig:pin_packed}
  \end{subfigure}
  \quad
  \begin{subfigure}{0.225\linewidth}
    \centering
    \includegraphics[width=1\columnwidth]{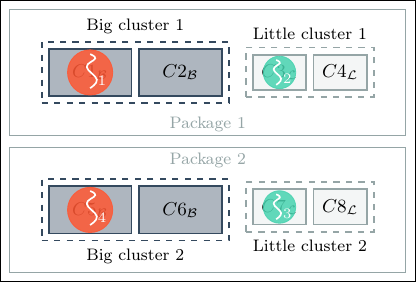}
    \caption{Distant policy.}
    \label{fig:pin_distant}
  \end{subfigure}
  \caption{Illustration of the evaluated pinning policies. Two clusters of two
    big cores ($C[(1,2),(5,6)]_\mathcal{B}$, dark gray plain boxes) and two
    clusters of two little cores ($C[(3,4),(7,8)]_\mathcal{L}$, light gray plain
    boxes) are considered. Big and little clusters 1 are regrouped into
    package~1, and big and little clusters 2 into package 2. In the examples,
    four threads need to be scheduled: thread 1 and 4 (red circles) represent a
    big amount of work and thread 2 and 3 (green circles) represent a small
    amount.}
  \label{fig:pin_policy}
\end{figure*}

\begin{figure*}[htp!]
  \begin{subfigure}[b]{0.225\linewidth}
    \centering
    \includegraphics[width=1\columnwidth]{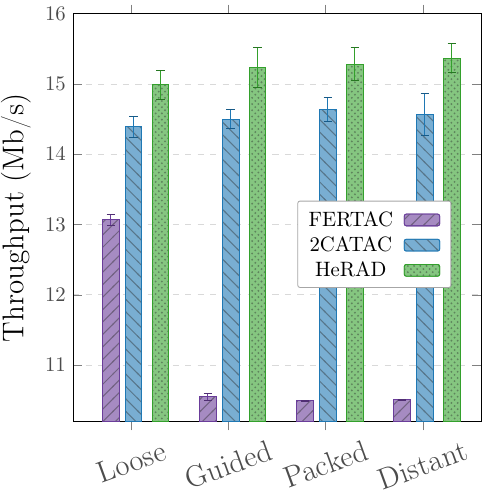}
    \caption{Orange Pi 5+.}
    \label{fig:pin_opi}
  \end{subfigure}
  \quad
  \begin{subfigure}[b]{0.225\linewidth}
    \centering
    \includegraphics[width=1\columnwidth]{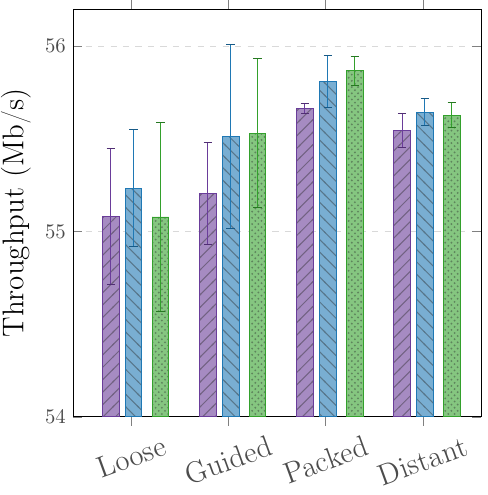}
    \caption{Mac Studio.}
    \label{fig:pin_m1u}
  \end{subfigure}
  \quad
  \begin{subfigure}[b]{0.225\linewidth}
    \centering
    \includegraphics[width=1\columnwidth]{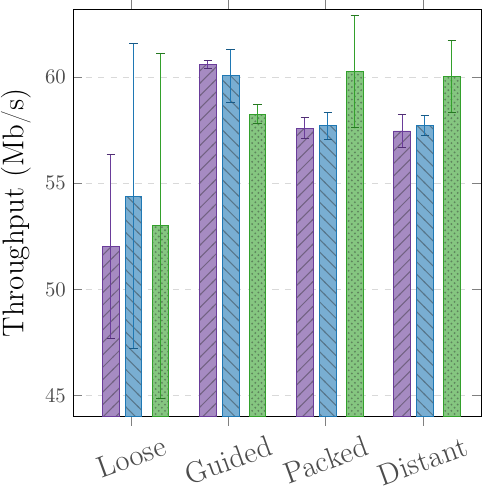}
    \caption{AI370.}
    \label{fig:pin_ai370}
  \end{subfigure}
  \quad
  \begin{subfigure}{0.225\linewidth}
    \centering
    \includegraphics[width=1\columnwidth]{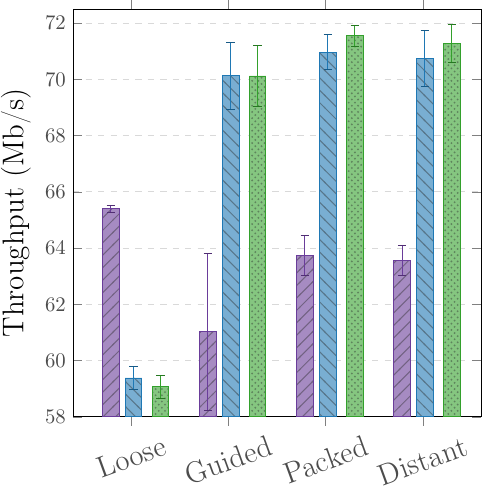}
    \caption{X7 Ti.}
    \label{fig:pin_x7ti}
  \end{subfigure}

  \caption{Throughput depending on the platform, the pinning policy and the
    heterogeneous scheduling algorithm. Each time, the solutions including all
    the cores are used. y-axis does not start from 0 and the scales are
    different for each plot. Min-max delimiters at the top of each bar
    represent the std dev. Higher is better.}
  \label{fig:pin_result}
\end{figure*}

Figure~\ref{fig:dvbs2_power} presents the instant power consumption (in Watts)
while running the DVB-S2’s receiver depending on the platform and the scheduling
strategy. This information, when combined with the throughput achieved on the
different scenarios, leads us to find the  energy consumed to compute one frame
of 14232 information bits, i.e., the energy efficiency of the combined platform
and scheduling algorithm. This is illustrated in Figure~\ref{fig:dvbs2_energy}.

\textbf{Orange Pi 5+}.
On this platform, the instant power consumed by the proposed schedules is mainly
correlated with their achieved throughputs. However, the energy efficiency of
\fertac ($\mathcal{S}_{8}$, 10.6~mJ), \twocatac ($\mathcal{S}_{9}$, 9.9~mJ), and
\herad ($\mathcal{S}_{10}$, 9.8~mJ) is significantly better than of \otac
($\mathcal{S}_{[6:7]}$, 17.5 and 11.8~mJ). It demonstrates that the
heterogeneous schedulers are at the same time the fastest and the most energy
efficient ones here.

\textbf{Mac Studio}.
\fertac, \twocatac, and \herad ($\mathcal{S}_{[18:20]}$, 11.5-11.7~mJ) show
similar energy efficiency. In the meantime, \otac[B] is less energy efficient
($\mathcal{S}_{17}$, 12.6~mJ), even though its throughput is similar to the
aforementioned strategies (55.9~Mb/s). This comes from its use of extra big
cores (instead of profiting from the more energy-efficient ones). This type of
result helps us validate our approach to minimize power consumption by using
\textit{as many little cores as necessary} (see Section~\ref{sec:def}).

\textbf{AI370}.
In this platform, \otac[L] appears as the most energy efficient solution
($\mathcal{S}_{26}$), consuming about 7.6~mJ per frame using only little cores.
Heterogeneous strategies, such as \herad ($\mathcal{S}_{30}$), use one extra big
core to achieve higher throughput (14\% higher, see Table~\ref{tab:dvbs2}) while
paying the price of a higher power demand (21\% higher), leading to a
consumption of 8.0~mJ per frame (8\% higher). In a situation where throughput
requirements may be lower, one could choose to use only little cores (with
\otac[L]) to improve energy efficiency.

\textbf{X7 Ti}.
Although the highest throughputs were achieved in this platform, it also showed
some of the worst energy consumption per frame. \otac[L] shows one of the best
energy efficiencies in the platform ($\mathcal{S}_{36}$, 16.7~mJ). It is
slightly outperformed by \twocatac and \herad ($\mathcal{S}_{[39:40]}$, 15.4 and
15.7~mJ). Differently from the AI370 platform, here the difference in energy
efficiency also comes with a big difference in throughput when using the
different heterogeneous resources, with strategies achieving throughput almost
three times higher (71.5~Mb/s) than \otac[L] (26.1~Mb/s).

\textbf{OS Strategies}.
The OS strategies provided disappointing results in terms of energy efficiency
in most scenarios. On all platforms, the instant power consumption increased
significantly with $R_2$ and $R_3$. Even more, this power increase usually came
with a decrease in throughput, leading to energy efficiency results far inferior
to the other strategies. This highlights the fact that the strategy of
generating several threads to improve efficiency is not working in these
platforms (at least with the Linux scheduler). The performance of custom OS
schedulers is studied later in Subsection~\ref{subsec:os_scx}.

\subsection{Pinning Policies} \label{subsec:pinning}

In this section, four different strategies for thread placement (also called
thread pinning/mapping) on the cores are studied. The pipeline decomposition in
stages is given by \fertac, \twocatac, and \herad. Then, each stage is run into
a thread. The placement of these threads onto cores is evaluated according to
the four policies described below.

\textbf{Loose}. This policy leaves the thread placement to the OS scheduler
(\textit{CFS} for Orange Pi 5 and \EEVDF for other platforms).
Nothing is done on the runtime layer and the OS scheduler can map the
threads wherever it wants (\emph{one to all} strategy). This is illustrated in
Figure~\ref{fig:pin_loose}.

\textbf{Guided}. In this policy, a thread is pinned to a set of cores (\emph{one
to many} strategy). If the stage has been made for little cores, the
corresponding set is all the available little cores. The same rule applies for
big cores. For instance, in Figure~\ref{fig:pin_guided}, thread~1 can be mapped
on cores $C[1,2,5,6]_\mathcal{B}$ but not on $C[3,4,7,8]_\mathcal{L}$. Thus,
thread~1 can be freely migrated on the big cores by the OS.

\textbf{Packed}. Each thread is pinned to a specific core (\emph{one to one}
strategy), leaving no freedom to the OS scheduler. In the packed strategy (also
sometimes referred to as \emph{compact} in the literature), the cores are
allocated in a fixed ascending order from the smallest to the biggest core
identifier. For instance, in Figure~\ref{fig:pin_packed}, thread 2 and 3 are
pinned to $C3_\mathcal{L}$ and $C4_\mathcal{L}$, respectively. These two cores
are on the same cluster, meaning that they are physically close to each other.
The distance between the core identifiers represents the locality between them.
For instance, $\dist(C1_\mathcal{B}, C3_\mathcal{L}) = 2$ and
$\dist(C2_\mathcal{B}, C6_\mathcal{B}) = 4$. In general, this strategy intends
to reduce the communication and synchronization overheads between consecutive
stages in the pipeline.

\textbf{Distant}. With this policy, each thread is pinned to a specific core
(\emph{one to one} strategy) but, contrary to packed, the distance between the
cores is maximized. For instance, in Figure~\ref{fig:pin_distant}, thread~1 is
pinned to the first core ($C1$) of cluster~1 and thread~4 is pinned to the first
core ($C5$) of cluster~2. In other words, a round-robin distribution over the
clusters of the same type is performed (same applies for packages if any). This
strategy is intended to increase the memory throughput as the caches are
generally replicated over the clusters and packages. However, it can slow down
the communication and synchronization between the threads of consecutive stages.

We have implemented these policies using \emph{hwloc}~\cite{broquedis2010hwloc}.
It provides a low level and hierarchical interface to pin the threads. Indeed,
hwloc generates core numbering based on \emph{core to core} latency, cache
hierarchy, and so on. This is packaged in our runtime starting from \spu v1.8.0.

Figure~\ref{fig:pin_result} presents the resulting throughput depending on the
pinning policies. In general, the packed and distant policies are giving the
best performances (and the lowest variations). We only considered the packed
policy in the previous sections as, in general, it limits the number of clusters
and packages used. Consequently, it minimizes the energy consumption.

\textbf{Orange Pi 5+}.
On this platform, guided, packed and distant policies perform similarly.
However, the loose policy results in a higher throughput with \fertac
($\mathcal{S}_8$ @ 13.0 Mb/s). This happens because the operating system can pin
the limiting stage 1 to big cores. Even if this is improving performance, this
solution is inferior compared to \twocatac ($\mathcal{S}_9$) and \herad
($\mathcal{S}_{10})$.

\textbf{Mac Studio}.
The loose and guided policies here perform worse than the packed and distant
ones on average. The fact that the operating system can freely migrate the
threads to different cores leads to high variations compared to the \emph{one to
one} strategies (context switches and data copies into the caches). That being
said, the differences in performance in this platform are somewhat minimal, with
all average throughputs falling in the range of 55 and 56~Mb/s.

\textbf{AI370}.
Like on the Mac Studio, the loose policy achieves lower performance with high
variations. On the other hand, the guided policy, combined to \fertac and
\twocatac ($\mathcal{S}_{[28:29]}$) gives some of the best results. It is
counterintuitive, but we think this is due to the fact that the guided
performance is not limited to the use of 7 resources (1 big and 6 little cores).
Even if the number of threads is fixed to 7, the OS is free to pin the threads
to more cores (\emph{one to many} strategy). In this particular case, it leads
to a better use of the resources (and maybe unlocks higher clock frequencies).
However, it is unclear if this relaxed strategy also has a positive impact on
energy efficiency.

\textbf{X7 Ti}.
For \fertac ($\mathcal{S}_{38}$), the loose pining policy outperforms the other
ones. This is primarily due to the fact that stages intended for the big cores
can be scheduled on the little cores, which severely degrades performance.
\fertac is the only solution to use all the 14 cores, and it includes 7~stages
in its pipeline. \twocatac ($\mathcal{S}_{39}$) and \herad ($\mathcal{S}_{40}$)
use 13~cores and have a 6-stage solution. This enables the OS to take advantage
of more resources by balancing the threads on the available cores. However, like
in Orange~Pi~5+, \fertac is clearly the worst solution when compared to
\twocatac and \herad, as they reach much higher throughput when combined with
the packed or distant policies (65.5~Mb/s for \fertac-loose versus 71.5~Mb/s for
\herad-packed).

\subsection{Investigating OS Scheduling} \label{subsec:os_scx}

\begin{figure}[h!]
  \begin{subfigure}{1.0\columnwidth}
  \centering
    \includegraphics[width=0.8\columnwidth]{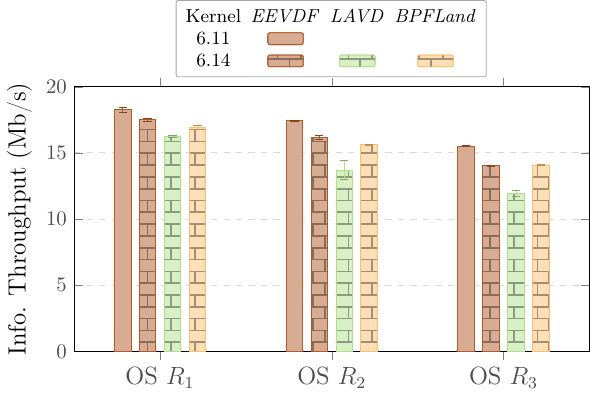}
  \caption{Throughput achieved across different OS schedulers. Higher is
    better.}
  \label{fig:sxc_thr}
  \end{subfigure}

  \vspace{0.4cm}
  \begin{subfigure}{1.0\columnwidth}
  \centering
    \includegraphics[width=0.8\columnwidth]{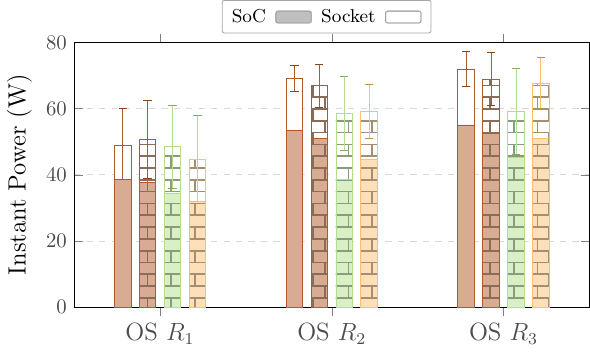}
    \caption{Instant power measured across different OS schedulers. Lower is
      better.}
    \label{fig:scx_power}
  \end{subfigure}

  \vspace{0.4cm}
  \begin{subfigure}{1.0\columnwidth}
  \centering
    \includegraphics[width=0.8\columnwidth]{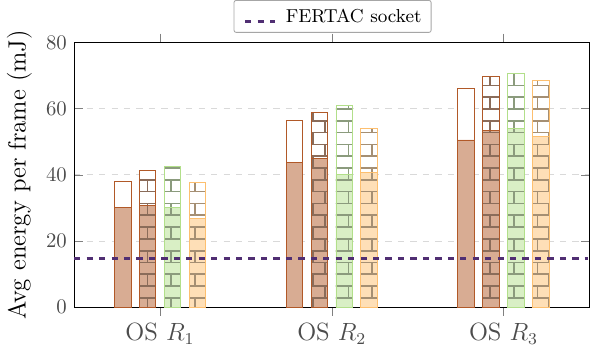}
    \caption{Energy per frame across different OS schedulers. Values are
      compared to \fertac, the least efficient static strategy. Lower is better.}
    \label{fig:scx_energy}
  \end{subfigure}
  \caption{Achieved performance on the DVB-S2 receiver for the OS strategies
    on the X7 Ti platform for the different system schedulers (\EEVDF is tested
    for both kernel versions $6.11$ and $6.14$).}
  \label{fig:scx}
\end{figure}

As we saw in the previous subsections, the OS strategies provided disappointing
results in most scenarios for the default Linux scheduler (\EEVDF or \CFS for
old Linux kernels). In this subsection, we propose to study the impact of
different OS schedulers on the DVB-S2 throughput and energy consumption.

Since Linux kernel $6.12$, a new framework, called
\schedext~\cite{hodges2024scheduling} (short for scheduler extension),
has been mainlined. It is a scheduler class whose behavior can be defined by a
set of BPF programs to replace the default scheduler or to be combined with it.
We decided to evaluate two of the most promising \schedext
implementations:
% \begin{inparaenum}[(i)]
\begin{itemize}
  \item \textit{Latency Aware Virtual Deadline} (\LAVD) proposed by Valve and
    used in portable gaming consoles where energy efficiency matters (typically
    used in SteamOS but not only). Additionally, it has also been recently used
    by Meta on its servers.
  \item \BPFLand proposed by an Nvidia system engineer and targeting low latency
    and high responsiveness. Consequently, it is mainly used in Linux
    distributions for gaming like CachyOS. However, we believe it is also
    relevant for the studied streaming applications where latency matters.
\end{itemize}
% \end{inparaenum}
\LAVD and \BPFLand are heterogeneous-aware and try to optimize energy efficiency
by building new CPU domains based on information collected in sysfs. To evaluate
and compare these OS schedulers, we chose to target the X7 Ti platform, where
\EEVDF (no extension enabled) performed poorly. To do so, we updated the kernel
to version $6.14$ (on Ubuntu 25.04).

As shown in Figure~\ref{fig:sxc_thr},
in terms of throughput, \EEVDF performs on average $15\%$ better than \LAVD
across the three strategies, and $2\%$ more than \BPFLand. We also observe
a slight difference between kernel versions $6.11$ (used in the previous
subsection) and $6.14$, where the older version achieves $5\%$ higher
performance. \EEVDF aims to maximizes instruction throughput by giving higher
priority to p-cores, while the \schedext schedulers focus more on
throughput-energy tradeoffs. As shown in Figure~\ref{fig:scx_power},
\BPFLand reduces power by $15\%$ compared to \EEVDF for $R_1$ and $R_2$
strategies, and by $2.2\%$ for $R_3$. For \LAVD, we observe a reduction of $7\%$
for $R_1$ and $15\%$ for $R_2$ and $R_3$. As a result,
Figure~\ref{fig:scx_energy} shows that \BPFLand is $10\%$ more energy efficient
than \EEVDF for $R_1$ and $R_2$, as the throughput degradation is balanced by
greater power savings. On the other side, \LAVD energy efficiency is not
improved compared to \EEVDF. This is mainly due to the fact that the throughput
degradation is more significant than for \BPFLand.

To summarize, the \schedext schedulers are interesting to investigate. As shown
earlier, \BPFLand slightly improves energy efficiency compared to \EEVDF.
However, in our experiments, we did not fully explore all the parameters
provided by \LAVD and \BPFLand. We believe their efficiency can be further
improved with fine tuning. Nevertheless, OS schedulers have limited knowledge
about the running application, which is why they cannot fully compete with the
proposed static strategies, as shown with \fertac baseline in
Figure~\ref{fig:scx_energy}.

\section{Concluding Remarks}\label{sec:conclusion}

\begin{figure}[htp]
  \centering
  \includegraphics[width=0.85\columnwidth]{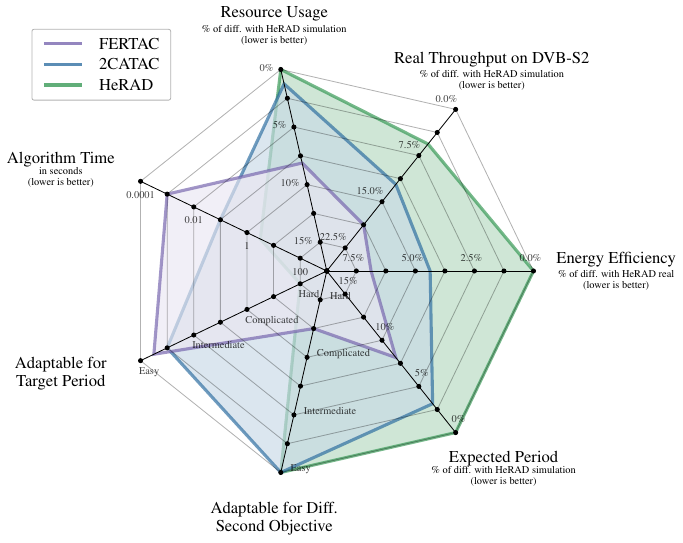}
  \caption{Advantages and limitations of the proposed strategies.}
  \label{fig:spider_algs}
\end{figure}

In this paper, we considered the problem of scheduling partially-replicable task
chains on heterogeneous multicore processors to optimize throughput (minimize
period) and power consumption (use as many little cores as necessary). We have
proposed two greedy strategies (\fertac that tries to use little cores as early
as possible, and \twocatac that tries to use both core types at each stage) and
one optimal dynamic programming strategy (\herad). Figure~\ref{fig:spider_algs}
summarizes their main characteristics based on our experimental evaluation and
analysis.

Using simulation, we have verified that \fertac and \twocatac are able to obtain
near-optimal schedules on average, with minor increases in period and resource
utilization. We have shown that the general quality of their schedules is
affected by characteristics of the platform (number of cores) and the task chain
(e.g., the number of replicable tasks).

Our real-world experiments with the DVB-S2 receiver task chain over four
heterogeneous multicore platforms have enabled us to validate our scheduling
strategies. Figure~\ref{fig:spider_algs} shows that, on average, the real
throughput difference compared to the best theoretical throughput (from \herad's
expected period) ranges between 12\% for \twocatac and 19\% for \fertac,
considering that \herad itself achieves 7\% differences of its target. We think
this as a positive result when moving from theory to practice. In terms of
throughput, our results have also emphasized the importance of using both types
of cores, as the optimal solution for homogeneous resources (\otac) usually
lagged behind our scheduling strategies. Considering the energy consumption, on
all the 40 tested solutions (except one on the AMD mini-PC), the heterogeneous
schedules are more energy efficient than \otac. This confirms our hypothesis
that using \emph{as many little cores as necessary} can be a win-win strategy
for both throughput and energy consumption. Considering \herad as a reference,
\fertac and \twocatac efficiency are only 7\% and 4\% lower, respectively.

As future work, we intend to take lessons from our experimental evaluation to
improve future solutions. We will profile the communication and synchronization
overheads on \spu to understand how they affect the schedules and include them
in our model, if necessary. We will study how to incorporate some of the
features of our best schedules (such as shorter pipelines) into our strategies.
We also plan to study frequency variations. We noticed that, when clock boost
technology is available, the difference between estimations and measurements is
significantly higher than others. Finally, we plan to evaluate the impact on
placing multiple stages on the same core to benefit from simultaneous
multithreading and very low communication overhead.

\section*{Acknowledgment}
{\small
\noindent
  This work has received support from France 2030 through the project named
  Académie Spatiale d'Île-de-France (\url{https://academiespatiale.fr/}) managed
  by the National Research Agency under bearing the reference ANR-23-CMAS-0041.
}

\bibliographystyle{elsarticle-num}
\bibliography{refs}

\end{document}